\documentclass[english,journal]{IEEEtran}
\usepackage[T1]{fontenc}
\usepackage[latin9]{inputenc}
\usepackage{array}
\usepackage{cprotect}
\usepackage{mathrsfs}
\usepackage{multirow}
\usepackage{dsfont}
\usepackage{amsmath}
\usepackage{amsthm}
\usepackage{amssymb}
\usepackage{graphicx}

\makeatletter

\providecommand{\tabularnewline}{\\}
\newcommand{\lyxdot}{.}

\theoremstyle{plain}
\newtheorem{lem}{\protect\lemmaname}
\theoremstyle{plain}
\newtheorem{thm}{\protect\theoremname}

\@ifundefined{date}{}{\date{}}
\usepackage{amsfonts}
\usepackage{cite}
\usepackage{array}
\usepackage{algorithm}
\usepackage{algorithmic}
\usepackage{subfigure}
\usepackage{stfloats}
\usepackage{gensymb}

\makeatother

\usepackage{babel}

\makeatother

\usepackage{babel}
\providecommand{\lemmaname}{Lemma}
\providecommand{\theoremname}{Theorem}

\begin{document}
\title{PoVD: Efficient Consensus Protocol\\
 based on Verifiable Delay Function}
\author{Rui~Jiang, Xintong~Ling,~\IEEEmembership{Member,~IEEE}, Bin
Cao,~\IEEEmembership{Senior Member,~IEEE},\\
Jiaheng~Wang,~\IEEEmembership{Senior Member,~IEEE}, Xiqi~Gao,~\IEEEmembership{Fellow,~IEEE},
Zhi~Ding,~\IEEEmembership{Fellow,~IEEE}\thanks{R. Jiang, X. Ling, J. Wang, and X. Gao are with the National Mobile
Communications Research Laboratory, Southeast University, Nanjing
210096, China, and also with the Purple Mountain Laboratories, Nanjing
210023, China (e-mail: \{ruijiang, xtling, jhwang, xqgao\}@seu.edu.cn).
J. Wang is also with School of Cyber Science and Engineering, Southeast
University, Nanjing 210096, China. Bin Cao is with the State Key Laboratory
of Networking and Switching Technology, Beijing University of Posts
and Telecommunications, Beijing 100876, China (e-mail: caobin@bupt.edu.cn).
Z. Ding is with Department of Electrical and Computer Engineering,
University of California, Davis, California, 95616 (e-mail: zding@ucdavis.edu).}}
\maketitle
\begin{abstract}
Consensus protocols ensure the robustness and scalability of blockchains
and decentralized applications built on them. However, existing consensus
mechanisms often impose high computational cost or require heavy communication
overhead. To address these challenges, we propose proof of verifiable
delay (PoVD), a lightweight consensus protocol based on the verifiable
delay function (VDF). We present the detailed protocol of PoVD, including
the block mining and verification rules, and illustrate how PoVD can
flexibly adjust the block time distribution according to the network
condition. Through mathematical proof, we point out that, under a
relatively weak network assumption, PoVD can achieve a lower fork
rate than PoW while maintaining the same throughput. Our experiments
verify that PoVD exhibits low computation and communication complexity
and can also improve the blockchain consistency, making it particularly
suitable for resource-constrained nodes and bandwidth-limited networks.
\end{abstract}


\section{Introduction}

With the rapid development of distributed applications, blockchains
have emerged as foundational infrastructures for supporting distributed
systems across various sectors, such as telecommunications, the Internet
of Things (IoT), and unmanned aerial vehicles (UAVs) \cite{Ling2025,Wu2025resource,cao2025optimization,Zou2021-2}.
As the core component of blockchain systems, the consensus protocol
is required to maintain network consistency among participating devices
and ensure the robustness and scalability of decentralized services
and applications \cite{Xu2023,Cui2025,Ling2025a,Wan2019}. However,
designing highly efficient consensus protocols for resource-constrained
and bandwidth-limited environments remains challenging. In such environments,
nodes often have limited resources and bandwidths and are unable to
handle high-frequency message exchanges or intensive computational
loads \cite{Chen2025,Zou2024}.

Generally, existing blockchain consensus protocols can be broadly
classified into two categories: interaction-based protocols that rely
on multi-round message exchange and proof-based protocols that often
require intensive computations. For example, practical Byzantine fault
tolerance (PBFT) is a well-known interaction-based protocol that achieves
agreement through four sequential stages, including request submission,
message propagation, verification, and commitment \cite{Castro1999}.
However, its $O\left(n^{2}\right)$ messaging complexity incurs significant
communication overhead \cite{Luo2024a}. Meanwhile, proof of work
(PoW) is the most recognized proof-based consensus protocol \cite{Nakamoto2008}.
In PoW, nodes compete in solving intensive cryptographic puzzles to
generate the next block, which consumes significant computational
resources.

As one can see, both types of protocols have their drawbacks. On the
one hand, interaction-based protocols require frequent multi-round
message exchanges, which leads to high communication complexity and
poor scalability for large networks. Especially, in bandwidth-limited
networks, such repeated message exchanges can increase block propagation
delay and raise the probability of forks, thereby undermining ledger
consistency. On the other hand, energy-intensive mining process of
proof-based protocols requires miners to use specialized hardware
with extremely huge power consumption, which not only raises sustainability
concerns but also limits more miners to participating in blockchain
maintenance. As a result, an efficient consensus problem must take
both communication overhead and computation cost into consideration.

The verifiable delay function (VDF) has the potential to address the
above key challenges. VDF is a cryptographic technique that cannot
be accelerated through parallel computing\cite{Wu2022}. Nodes cannot
predict the result in advance but have to compute for a certain period
of time to obtain the solution to VDF. That is, VDF can ensure that
each node experiences a verifiable delay without consuming too much
computing power, which is applicable for resource-constrained devices.
Moreover, the output of the VDF can be efficiently verified, which
can avoid multi-round interactions or a large amount of computations.
Therefore, it is promising to use VDF to design advanced consensus
protocols to reduce both communication overheads and computation consumption.

\subsection{Related Work}

Most existing lightweight consensus protocols are permissioned and
based on message exchanges. Recent studies in \cite{Luo2024a,Li2020TPDS,Mohammad2023Hotstuff,Cheng2025JUMBO}
aim to lower the consensus overhead of BFT-style consensus by, e.g.,
localizing commit stage communication \cite{Li2020TPDS}, streamlining
the pre-prepare and prepare pipelines \cite{Mohammad2023Hotstuff},
or reducing the cost of verifying quorum certificates \cite{Cheng2025JUMBO}.
RAFT-based consensus protocols attempt to improve the heartbeat mechanism
by embedding useful information into heartbeat messages \cite{Xu2021}
or reducing the frequency of heartbeat transmissions \cite{Luo2024}.

Several works adopted a simpler approach to reduce consensus overhead
via leader election and lightweight agreement workflows. \cite{Islam2025PoRL,Zhao2024Context}
used verifiable random function to simplify proposer selection. The
authors of \cite{Zhai2024MultiLeader,Zhai2024AccountCommit} proposed
committee-election mechanisms, which confine voting to small subsets
to limit messaging. \cite{Fu2022VaaP} employed a vote-based proof
mechanism to simplify the proposal and confirmation workflow. \cite{Zhou2024Accelerate}
introduced decentralized coordinator services with verifiable global
states. \cite{Hao2024BitFT} designed a resource-efficient consensus
protocol that combines multi-round sortition with vote-based confirmation.

\begin{table*}
\caption{Related lightweight consensus protocols. \label{tab:Related works}}
\scriptsize\centering\scalebox{1.1}{\renewcommand\arraystretch{1.3}

\begin{tabular}{|c|c|c|c|c|c|c|}
\hline 
\multirow{2}{*}{Ref.} & \multirow{2}{*}{Name} & \multirow{2}{*}{Year} & \multirow{2}{*}{Type} & \multirow{2}{*}{Design Focus} & \multirow{2}{*}{Role of VDF} & \multirow{2}{*}{Protocol Framework}\tabularnewline
 &  &  &  &  &  & \tabularnewline
\hline 
\hline 
\cite{Long2019} & / & 2019 & Permissionless & Resource consumption & Replace PoW & \multicolumn{1}{c|}{PoW-fashion}\tabularnewline
\hline 
\cite{Han2020} & RandChain & 2020 & Permissionless & Communication overhead & Replace PoW & PoW-fashion\tabularnewline
\hline 
\cite{Li2020TPDS} & / & 2020 & Permissioned & Communication overhead & Not used & PBFT\tabularnewline
\hline 
\cite{Deb2021} & PoSAT & 2021 & Permissioned & Dynamic availability consensus & Random beacon & Leader selection\tabularnewline
\hline 
\cite{Raikwar2021} & R3V & 2021 & Permissioned & Communication complexity and fairness & Replace PoW & PoW-fashion\tabularnewline
\hline 
\cite{Xu2021} & Weighted RAFT & 2021 & Permissioned & Communication latency & Not used & RAFT\tabularnewline
\hline 
\cite{Fu2022VaaP} & VaaP & 2022 & Permissioned & Communication overhead & Not used & BFT\tabularnewline
\hline 
\cite{Raikwar2022} & / & 2022 & Permissioned & Fairness & Random beacon & Leader election\tabularnewline
\hline 
\cite{Xu2022-2} & Fairledger & 2022 & Permissioned & Resource consumption and fairness & Replace PoW & PoW-fashion\tabularnewline
\hline 
\cite{Mohammad2023Hotstuff} & Fast-HotStuff & 2023 & Permissioned & Communication overhead & Not used & \multicolumn{1}{c|}{BFT}\tabularnewline
\hline 
\cite{Li2023} & BLMA & 2023 & Permissioned & Resource consumption and communication latency & Not used & \multicolumn{1}{c|}{PBFT}\tabularnewline
\hline 
\cite{pu2023} & Gorilla & 2023 & Permissioned & Protocol operation efficiency and safety & Random beacon & \multicolumn{1}{c|}{BFT}\tabularnewline
\hline 
\cite{Zhai2024MultiLeader} & / & 2024 & Permissioned & Fairness & Not used & Leader election\tabularnewline
\hline 
\cite{Zhai2024AccountCommit} & / & 2024 & Permissioned & Protocol operation efficiency and safety & Not used & Leader election\tabularnewline
\hline 
\cite{Zhou2024Accelerate} & / & 2024 & Permissioned & Communication latency & Not used & BFT\tabularnewline
\hline 
\cite{Hao2024BitFT} & BitFT & 2024 & Permissioned & Resource consumption & Not used & BFT\tabularnewline
\hline 
\cite{Zhao2024Context} & / & 2024 & Permissioned & Protocol operation efficiency and safety & Not used & BFT\tabularnewline
\hline 
\cite{Wang2024} & / & 2024 & Permissionless & Protocol operation efficiency and safety & Random beacon & Leader election\tabularnewline
\hline 
\cite{Mirkin2024} & Sprints & 2024 & Permissionless & Resource consumption & Replace PoW & PoW-fashion\tabularnewline
\hline 
\cite{Das2024} & / & 2024 & Permissioned & Protocol operation efficiency and safety & Random beacon & BFT\tabularnewline
\hline 
\cite{Luo2024a} & GS & 2024 & Permissioned & Resource consumption & Not used & PBFT\tabularnewline
\hline 
\cite{Luo2024} & SBC & 2024 & Permissioned & Resource consumption and communication latency & Not used & RAFT\tabularnewline
\hline 
\cite{Xiong2025} & PoVF & 2025 & Permissioned & Protocol operation efficiency & Random beacon & Leader election\tabularnewline
\hline 
\cite{Islam2025PoRL} & PoRL & 2025 & Permissioned & Protocol operation efficiency and safety & Not used & Leader election\tabularnewline
\hline 
\cite{Cheng2025JUMBO} & JUMBO & 2025 & Permissioned & Resource consumption and communication latency & Not used & BFT\tabularnewline
\hline 
\end{tabular}}
\vspace{-0.3cm}
\end{table*}

Most recently, VDF is introduced in consensus protocols as a cutting-edge
technology \cite{Boneh2024}; however, it has not been well utilized
yet. In general, VDF plays two roles in consensus protocols. First,
VDF is used as the random beacon for leader or committee member elections.
For example, \cite{Deb2021,Xu2022-2,pu2023,Raikwar2022} leveraged
the unpredictability of VDF outputs to elect leaders for generating
new blocks. \cite{Xiong2025,Wang2024} utilized VDF as a random timer,
allowing only committee members within designated time windows to
propose new blocks. In \cite{Das2024}, VDF was introduced to support
fixed-time polling in Byzantine agreement (BA) protocols and promoted
election fairness and efficiency in the consensus process. In such
consensus protocols, the core idea of introducing VDF is to improve
the reliability of the election results.

As another approach, VDF is used as a mining puzzle to replace the
hash function. \cite{Long2019,Han2020} are two typical examples where
the first miner solving the VDF puzzle wins the right to create a
new block. Some, such as Sprints \cite{Mirkin2024}, use the same
principle but associated VDF-based mining with hash puzzles. These
protocols \cite{Long2019,Mirkin2024,Han2020} did not prevent miners
from performing parallel computations with different VDF inputs. As
a result, miners may attempt to accelerate block generation by investing
computational power, leading to high energy consumption. The characteristics
of these schemes are similar to PoW. To resist parallel computing,
R3V \cite{Raikwar2021} limited the VDF input and excluded the block
hash from the mining process. However, without involving the block
hash into the mining process, such a consensus protocol cannot protect
the on-chain information from being tampered with. 

Table \ref{tab:Related works} summarizes the above related studies
on consensus protocols. We can identify two fundamental drawbacks
of existing works from Table \ref{tab:Related works}. Most lightweight
consensus protocols without using VDF are often based on intensive
interactions. Although some designs attempt to reduce communication
complexity, they still require repeated broadcasting and flooding,
imposing a heavy burden on the network. Particularly, as the network
size grows, these schemes are not scalable because of rapidly increasing
communication overhead. As a result, some works have to set the scenario
within single-hop or single-cell networks, which are somehow centralized.

Truly, VDF can largely reduce the amount of message exchange. However,
the role of VDF is still questionable. The existing consensus protocols
based on VDF face the dilemma between computational efficiency and
blockchain security. From the above review, one can see that the dilemma
comes from how the protocol associates the block content or its digest,
such as the block hash, with the VDF input. The security and integrity
of most blockchains are based on the connections between blocks based
on the block hashes to protect the on-chain data from being modified.
However, when mining, miners can change the block content (e.g., change
the order of transactions) and generate multiple block hashes so that
they can compute with multiple VDF inputs in parallel. In this way,
miners can gain extra advantages by investing additional computing
power, resulting in a similar energy consumption problem as PoW does.
However, if we remove the block content (or its digest) from the VDF
computing, then the new chain structure cannot guarantee the integrity
of on-chain information. The above dilemma may be one of the important
reasons why VDF has not been widely used for consensus protocols. 

\subsection{Our Contributions}

In this work, we propose an efficient consensus protocol based on
VDF, named proof of verifiable delay (PoVD), which is both communication-
and computation-efficient and suitable for deployment in resource-constrained
and bandwidth-limited environments. PoVD does not require intensive
message exchange and can operate in a permissionless setting. PoVD
introduces VDF to resist parallel computing so that miners cannot
gain advantages by increasing computational resources, and thus, PoVD
can effectively reduce energy consumption. Furthermore, PoVD can flexibly
adjust the block time distribution so that it can achieve better chain
performance according to the network condition. Our experiments show
that PoVD is suitable for resource-constrained devices in bandwidth-limited
networks. We summarize our main contributions as follows:
\begin{itemize}
\item We design the PoVD consensus protocol with the blockchain structure
and illustrate the entire consensus process in detail, including VDF-based
mining, VDF proof generation, and block verification.
\item Since PoVD can adjust the block time distribution, we build a mining
model under arbitrary block time distributions and assess the performance
of PoVD in terms of fork rate and blockchain throughput.
\item Through rigorous mathematical proof, we point out that, under a relatively
weak network condition, PoVD can achieve a lower fork rate than PoW
with the same blockchain throughput.
\item By comparing with several representative benchmarks, we demonstrate
PoVD\textquoteright s low computational and communication complexities
and highlight PoVD\textquoteright s advantage in improving blockchain
consistency. The code is available at: https://github.com/RayJ9/VDF-based-PoVD.
\end{itemize}
The remainder of the paper is structured as follows. Section \ref{sec:System Model}
presents the system model. Section \ref{sec:Verifiable Delay Function}
describes the basic principles of VDF. Section \ref{sec:Consensus-Protocol}
illustrates the PoVD protocol. Section \ref{sec:Protocol Modeling}
models the PoVD mining process. Section \ref{sec:Protocol Analysis}
shows the superiority of PoVD regarding fork rate. Section \ref{sec:Simulation}
presents experimental results, and Section \ref{sec:Conclusions}
concludes the paper.

\section{\label{sec:System Model}System Model}

\subsection{Byzantine Agreement Setting}

Our study is based on the BA setting \cite{Ni2020,Garay2024,Das2024},
which defines a system with a finite number of participants and models
the system's operation by dividing real-world continuous time into
discrete rounds of equal duration. All nodes operate under the same
constraints in each round. The discrete BA setting can approach the
accurate continuous-time performance with a shorter round duration.

Specifically, in BA setting, the system comprises $n$ miners. These
miners attempt to generate new blocks by computing random oracles
in every round. For example, in PoW, the random oracle is the hash
function such as SHA256. In general, if a miner obtains a valid oracle
solution, the miner wins the right to generate a new block.

The new block will be broadcast to the entire network after being
generated. The rest of the miners will verify it based on the verification
criteria of the consensus protocol. If the block is verified, miners
will add it to their local chains to synchronize the state. Usually,
the verification time is negligible compared to the time spent generating
new blocks \cite{Das2024}. 

Furthermore, we assume a flat model where all miners possess identical
computational capabilities, i.e., every miner performs the same number
of oracles per round \cite{Ni2020}. Even though this assumption may
not hold in practice, we can cluster the flat-model miners into larger
virtual entities each of which comprises more than one flat-model
player to describe the practical non-flat case where miners have differing
computational capabilities. For example, if there are 10 miners, with
half possessing twice the computing power of the others, we can create
a new model of 15 miners, where 10 pair up to form a new unit. This
adjustment allows the analysis to proceed as though there are 15 equivalent
miners.

\subsection{\label{subsec:Proof-of-Work}PoW in BA Setting}

Next, we use PoW as a typical example to illustrate the above BA setting.
Miner $i$ continuously attempts different inputs, i.e., nonce, to
calculate hash values satisfying the difficult target:
\begin{equation}
y_{h,i}=\mathrm{Hash}\left(y_{h-1},\mathrm{ID}_{i},m_{h,i},\pi_{h,i}\right)<\varphi,
\end{equation}
where $h$ is the blockchain height, $y_{h-1}$ represents the hash
of the previous block, $\textrm{ID}_{i}$ is the miner's identity,
$m_{h,i}$ is the block payload data represented by its digest, $\pi_{h,i}$
denotes the nonce, and $\varphi$ indicates the current hash target.
When a miner finds a nonce whose hash value is smaller than the target
$\varphi$, the miner generates a new block with the corresponding
hash value and broadcasts it to the other miners. The newly generated
block will be added into the main chain if there is no other fork,
and the blockchain grows in this way. The hash target $\varphi$ is
determined by the mining difficulty to control the average block time.
For example, a larger hash target implies an easier mining process
and thus a shorter block time. We show the workflow of PoW in Fig.
\ref{fig:Workflow of the PoW} of Section IV.

\subsection{Network Model}

Usually, block propagation inevitably experiences delays, particularly
in limited network environments. To characterize the propagation
process, we introduce a $d$-dimensional propagation vector:
\begin{align}
\mathbf{w}_{d}= & \left(w_{0},w_{1},\ldots,w_{d-1}\right),
\end{align}
where $d$ is the maximum propagation delay. The element $w_{i}\in\left[0,1\right)$
for $i=0,1,...,d-1$ represents that $nw_{i}$ miners receive the
block in round $r+i$ if the block is generated in round $r$. Note
that, if a block is generated in round $r$, $n-1$ miners are not
yet aware of the block in round $r$, implying that $w_{0}$ is always
$1/n$. During the propagation process, more miners will receive the
block, so $w_{0}\leq w_{1}\leq\cdots\leq w_{d-1}<1$. Given the maximum
propagation delay $d$, all the miners will receive the block in round
$r+d$. In an ergodic and stable network, $\mathbf{w}_{d}$ can be
statistically measured through experimental results. Furthermore,
we define the network capability indicator $c$ as
\begin{equation}
c=\frac{1}{d}\sum_{i=0}^{d-1}w_{i},\label{eq:network value}
\end{equation}
which can reflect the network's propagation capability. Given the
maximum delay $d$, a larger $c$ implies more miners receive the
block earlier, i.e., a faster propagation speed. Note that the network
capability indicators $c$ are incomparable for two networks with
different $d$. Network propagation model does not involve the number
of miners and thus can provide convenience for subsequent mathematical
modeling.

\section{\label{sec:Verifiable Delay Function}Verifiable Delay Function}

\subsection{Definition}

This section will provide a detailed definition of VDF and explain
its principle \cite{Boneh2024}. A verifiable delay function is a
function $f$: $\mathds{\mathbb{Z}}\rightarrow\mathscr{\mathbb{Y}}$
that requires a predetermined amount of time to compute and cannot
be accelerated by parallel processing. Every input $z\in\mathds{\mathbb{Z}}$
has a valid output $y\in\mathscr{\mathbb{Y}}$. Moreover, once the
computation is complete, anyone can verify the output quickly. 

More specifically, a VDF that implements a function $\mathds{\mathbb{Z}}\rightarrow\mathscr{\mathbb{Y}}$
is a tuple of three algorithms: 
\begin{itemize}
\item $\mathrm{Setup}\left(\mathsf{sp},T\right)\rightarrow\mathbf{\mathsf{pp}}$
is a randomized algorithm that takes a security parameter $\mathsf{sp}$
and a time bound $T$, and outputs a public parameter $\mathbf{\mathsf{pp}}$.
\item $\mathrm{Eval}\left(\mathsf{pp},z\right)\rightarrow\left(y,\pi\right)$
takes an input $z\in\mathds{\mathbb{Z}}$ and outputs the VDF result
$y\in\mathscr{\mathbb{Y}}$ and the proof $\pi$.
\item $\mathrm{Verify}\left(\mathsf{pp},z,y,\pi\right)\rightarrow\left\{ \mathrm{accept};\mathrm{reject}\right\} $
outputs accept if $y$ is the correct evaluation of the VDF on input
$z$, or outputs reject if $y$ is incorrect.
\end{itemize}
Meanwhile, a VDF must satisfy three properties:
\begin{itemize}
\item Bounded-evaluation time. $\mathrm{Eval}\left(\mathbf{\mathsf{pp}},z\right)\rightarrow\left(y,\pi\right)$
runs no less than $T$ time, for all $z\in\mathds{\mathbb{Z}}$ and
all $\mathsf{pp}$ output by $\textrm{Setup}\left(\mathsf{sp},T\right)$. 
\item Serialization. A parallel algorithm $\mathscr{A}$ that uses at most
$\textrm{poly}\left(\mathsf{sp}\right)$ processors and runs in time
less than $T$ cannot compute the function. Specifically, for a random
$z\in\mathds{\mathbb{Z}}$ and $\mathbf{\mathsf{pp}}$ output by $\mathrm{Setup}\left(\mathsf{sp},T\right)$,
if the correct result of $\textrm{Eval}\left(\mathbf{\mathsf{pp}},z\right)$
is $\left(y,\pi\right)$ then $\textrm{Pr}\left(\mathscr{A}\left(\mathbf{\mathsf{pp}},z\right)=y\right)$
is negligible. 
\item Uniqueness. For a given input $z\in\mathds{\mathbb{Z}}$, exactly
one $y\in\mathscr{\mathbb{Y}}$ will be accepted. Specifically, given
$\mathsf{pp}$ as an input, let $\mathscr{A}$ be an efficient algorithm
that outputs $\left(z,y,\pi\right)$ such that $\textrm{Verify}\left(\mathbf{\mathsf{pp}},z,y,\pi\right)=\mathrm{accept}$.
Then $\textrm{Pr}\left(\mathscr{A}\left(\mathbf{\mathsf{pp}},z\right)\neq y\right)$
is negligible.
\end{itemize}

\subsection{Wesolowski's Implementation}

Currently, there are several implementation approaches to VDF. This
subsection describes Wesolowski's scheme \cite{Boneh2024}, which
serves as the foundation for our proposed PoVD. Specifically, the
input consists of a tuple $\left(\mathbb{G},z,y,T\right)$, where
$\mathbb{G}$ denotes an integer group with finite elements. A prover
and a verifier are included in the following protocol to prove $y=z^{\left(2^{T}\right)}\in\mathbb{G}$.
Let $\textrm{Primes}\left(\kappa\right)$ be the set containing the
first $2^{\kappa}$ primes.
\begin{enumerate}
\item The verifier checks whether $z,y\in\mathbb{G}$ or not.
\item The verifier sends to the prover a random prime $v$ sampled uniformly
from $\textrm{Primes}(\kappa)$.
\item The prover computes $u,\varepsilon\in\mathbb{Z}_{++}$ such that $2^{T}=uv+\varepsilon$
with $0\leq\varepsilon<v$, and sends $\pi\leftarrow z^{u}$ to the
verifier. Here, $\mathbb{Z}_{++}$ denotes the positive integer set.
\item The verifier computes $\varepsilon\leftarrow2^{T}$ mod $v$ and outputs
$\mathrm{accept}$ if $\pi\in\mathbb{G}$ and $y=\pi^{v}z^{\varepsilon}\in\mathbb{G}$.
\end{enumerate}
The verifier computes $\varepsilon\leftarrow2^{T}$ mod $v$, which
only takes $\log_{2}T$ multiplications in $\mathbb{Z}_{++}/v$. In
\cite{Boneh2024}, a constant space algorithm for efficiently computing
$\pi=z^{u}\in\mathbb{G}$ is implemented. Hence, the verifier is efficient.

\section{\label{sec:Consensus-Protocol}PoVD Consensus Protocol}

\subsection{Overview}

We first illustrate the block structure of PoVD and how blocks connect
to one another in Fig. \ref{fig:Blockchain Structure.}. Every block
includes the block proposer's ID, the last block's VDF result $y_{h-1}$,
the current block's payload data $m_{h}$ (recorded by the hash value),
the VDF result $y_{h}$, and the proof $\pi_{h}$. Similarly to PoW,
PoVD also has the mining and verification phases, as shown in Fig.
\ref{fig:VDF-based consensus protocol}. (For comparison, we also
show the workflow of PoW in Fig. \ref{fig:Workflow of the PoW}.)
Specifically, in the PoVD mining process, miners compute the VDF,
generate the necessary proofs, and encapsulate them into the proposed
blocks. The miner who first solves the VDF puzzle earns the right
to propose a new block. Once a new block is proposed, other miners
in the network verify it, i.e., the verification phase. They verify
the VDF solution's correctness and proof and check the block's consistency
with the previous block. If a block passes all the verification, it
is appended to the miner's local blockchain. The blockchain then grows
in this manner.
\begin{figure}[t]
\centering{}\includegraphics[width=0.5\textwidth]{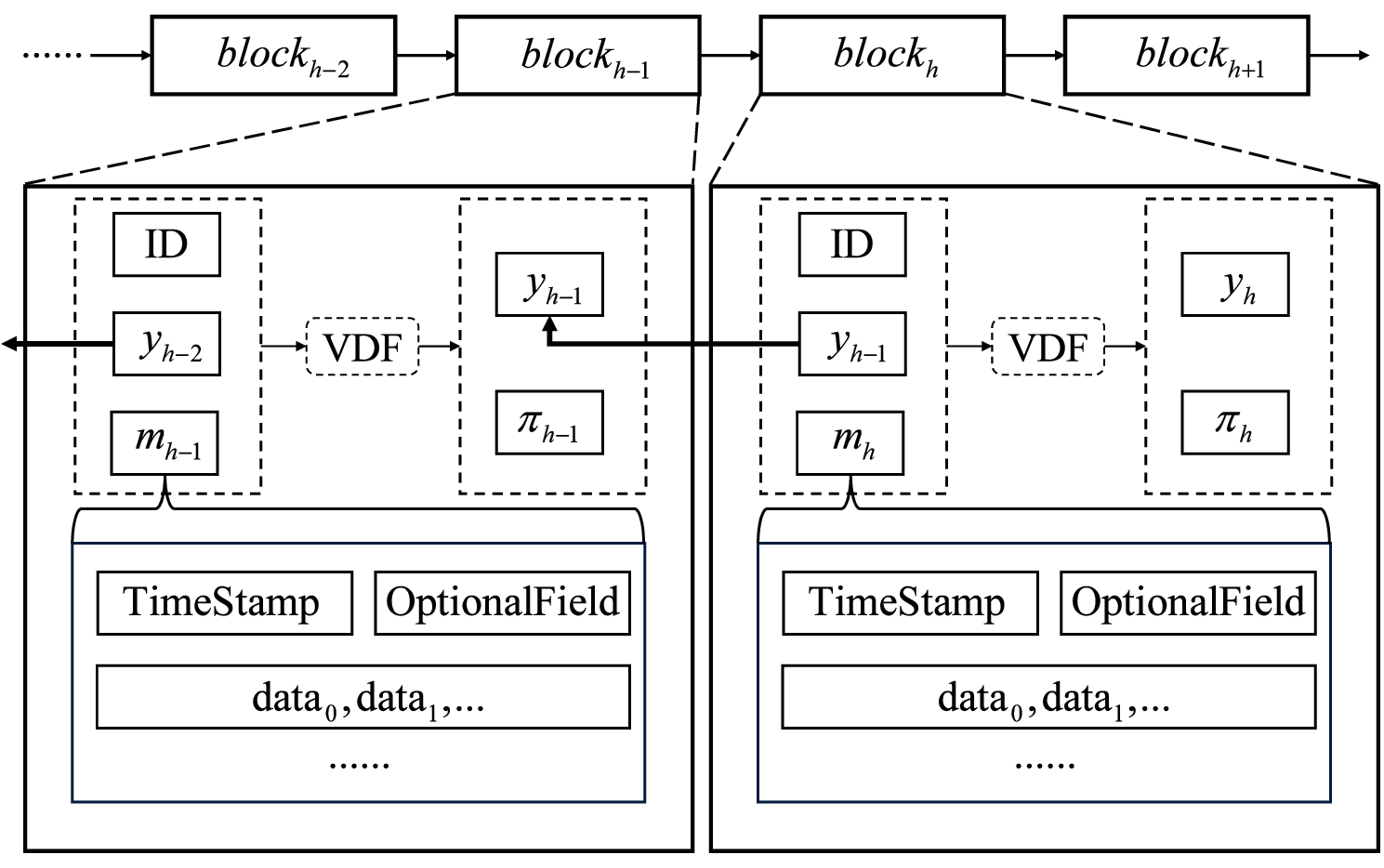}\caption{Illustration of PoVD block structure. \label{fig:Blockchain Structure.}}
\vspace{-0.6cm}
\end{figure}
\begin{figure}[t]
\centering{}\subfigure[]{\includegraphics[width=0.4\textwidth]{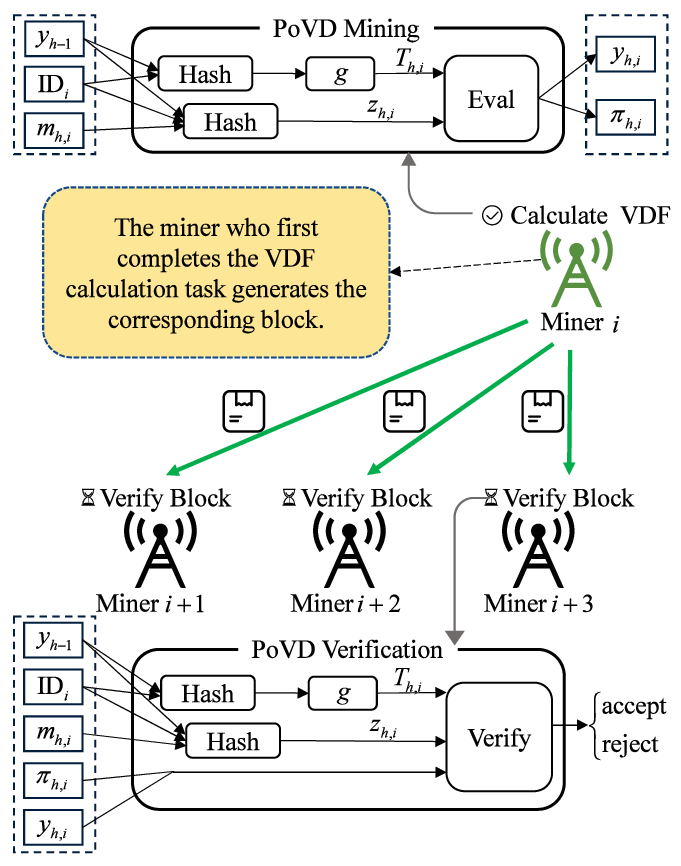}
\label{fig:VDF-based consensus protocol}}\hfill{}\subfigure[]{\includegraphics[width=0.4\textwidth]{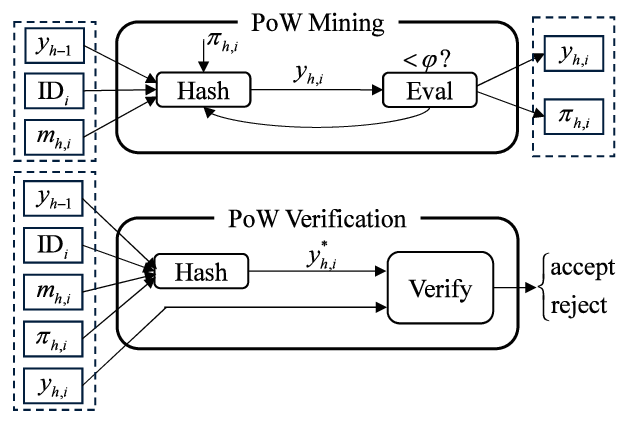}\label{fig:Workflow of the PoW}}\caption{Block mining process and verification rules of PoVD and PoW consensus
protocols. (a) PoVD. (b) PoW.}
\vspace{-0.6cm}
\end{figure}

\subsection{\label{subsec:PoVD-Mining}PoVD Mining}

Consider a blockchain with height $h-1$. For miner $i$ attempting
to generate a block at height $h$, the VDF mining process begins
with generating the input $z_{h,i}$ and time parameter $T_{h,i}$:
\begin{align}
z_{h,i} & =\textrm{Hash}\left(y_{h-1},\textrm{ID}_{i},m_{h,i}\right),\label{eq:x gener}\\
T_{h,i} & =g\left(\textrm{Hash}\left(y_{h-1},\textrm{ID}_{i}\right)\right).\label{eq:T gener}
\end{align}
Here, $y_{h-1}$ is the VDF result from the previous block at height
$h-1$, which is usually the same for all miners in the system, and
therefore, we omit the subscript $i$ representing the winner of the
last block. $\textrm{ID}_{i}$ is the miner's identity, and $m_{h,i}$
is the block hash representing the block payload data. Meanwhile,
the calculation of $z_{h,i}$ and $T_{h,i}$ requires two hash functions:
\begin{itemize}
\item $\textrm{Hash}$: $\mathbb{S}\rightarrow\mathbb{C}$, which is a standard
cryptographic hash function, with $\mathbb{S}$ representing the set
of strings of arbitrary length and $\mathbb{C}$ denoting a set of
integers. This function maps $\mathbb{S}$ uniformly into $\mathbb{C}$. 
\item $g$: $\mathbb{C}\rightarrow\mathbb{G}$, which maps integers from
$\mathbb{C}$ to another integer group $\mathbb{G}$ with predefined
settings.
\end{itemize}
The number and the range of elements in the set $\mathbb{C}$ depend
on the function $\textrm{Hash}$. For example, in the widely used
SHA-256, $\mathbb{C}$ is the set of integers from 0 to $2^{256}-1$,
with each element having an identical probability of being selected.
For any arbitrary outputs $x$, its values are uniformly distributed
among the set $\mathbb{C}$, i.e., 
\begin{equation}
\textrm{Pr}\left(x=\textrm{Hash}\left(s\right)\right)=\frac{1}{\left|\mathbb{C}\right|}.\label{eq:Hash distribution}
\end{equation}
That is, in \eqref{eq:T gener}, the output of the function $\textrm{Hash}$,
which is also the input of $g$, follows a uniform distribution.

In PoVD, miner $i$ iteratively computes the square of $z_{h,i}$
for $T_{h,i}$ times in $\mathbb{G}$ to derive $y_{h,i}$ by the
fast power algorithm, yielding the VDF result:
\begin{align}
y_{h,i}=z_{h,i}^{2^{T_{h,i}}}= & \left(\textrm{Hash}\left(y_{h-1},\textrm{ID}_{i},m_{h,i}\right)\right)^{2^{g\left(\textrm{Hash}\left(y_{h-1},\textrm{ID}_{i}\right)\right)}}.
\end{align}
 Miner $i$ must complete $T_{h,i}$ rounds' computation to generate
a valid VDF result (if we assume that each miner computes one square
per round). Hence, $T_{h,i}$ is the time to generate a new block,
i.e., the so-called verifiable delay in PoVD. The function $g$ can
control the distribution of $T_{h,i}$, i.e., the block time distribution,
given the uniformly distributed Hash output shown by \eqref{eq:Hash distribution}.
That is how PoVD adjusts the block time distribution. 

Upon obtaining the VDF result $y_{h,i}$, miner $i$ generates the
proof $\pi_{h,i}$ according to the following steps:
\begin{enumerate}
\item Determine a factor $v$ of the decomposed exponent. A predetermined
rule is used to select one of the factors of the exponent. For example,
$v$ can be the prime indexed by $\textrm{Hash}\left(y_{h,i},z_{h,i},T_{h,i}\right)$
within the set $\textrm{Primes}(\kappa)$.
\item Find another exponent $u$ satisfying $2^{T_{h,i}}=uv+\varepsilon$,
where $0\leq\varepsilon<v$.
\item Generate the proof $\pi_{h,i}=z_{h,i}^{u}$.
\end{enumerate}
Usually, the time to generate the proof is negligible. Once miner
$i$ completes these steps, it can directly package the computed results
into a block and broadcast the block to the network. The consensus
process guarantees that miner $i$ must experience the delay $T_{h,i}$
with a verifiable proof. That is why the consensus protocol is called
proof of verifiable delay.

By comparing Fig. \ref{fig:VDF-based consensus protocol} and Fig.
\ref{fig:Workflow of the PoW}, we can find the similarities and differences
between PoW and PoVD. Essentially, both PoW and PoVD are races among
miners. Both protocols require miners to generate a new block by finding
a solution satisfying certain conditions. Such a process automatically
embeds a period of time into the block generation process to reduce
the probability that multiple new blocks are generated simultaneously,
which would result in forks. In PoW, a miner can perform hash trials
in parallel to increase the winning probability. As a result, PoW
consumes immense computational power and often requires specific hardware
such as mining rigs. In PoVD, miners cannot gain extra payoff by investing
more in computational resources since the time to generate a new block
cannot be accelerated by parallel computing. Hence, PoVD is more energy
efficient than PoW since the race on computational resources becomes
meaningless. 

\subsection{PoVD Verification}

In this subsection, we show the verification rule of PoVD. As shown
in Fig. \ref{fig:Blockchain Structure.}, the PoVD verification requires
the following elements: the miner $\textrm{ID}$, the calculation
input of the current block $z_{h}$, the time parameter $T_{h}$,
the proof $\pi_{h}$, the calculation result $y_{h}$ and the calculation
result of the previous block $y_{h-1}$. If any of them is absent,
the block should be considered illegal. Except for $z_{h}$ and $T_{h}$,
all other elements can be directly obtained from the block. Every
miner can derive $z_{h}$ and $T_{h}$ through \eqref{eq:x gener}
and \eqref{eq:T gener}, respectively. 

Based on these elements, miners can further calculate $v$ and $\varepsilon$.
As mentioned in Section \ref{subsec:PoVD-Mining}, the method for
generating $v$ is predetermined and known to all miners. Subsequently,
miners derive $\varepsilon$ by performing the modulus operation on
$v$ with $2^{T}$. Finally, they calculate $y_{h}^{*}=\pi^{v}z^{\mathcal{\varepsilon}}$
and compare it with $y_{h}$. If the two values match, the block passes
the verification.

We also need to verify the chain consistency. The consistency verification
requires the previous block's computation result to be included in
the currently received block's header, denoted as $y_{h-1}^{*}$.
This value is also in the previous block's header as $y_{h-1}$. If
these two values are the same, the block passes the consistency check.
From the description in the above section, $y_{h-1}$ is derived from
$z_{h-1}$, and the generation of $z_{h-1}$ depends on the previous
block. Therefore, the consistency of $y_{h-1}$ directly determines
whether the data between consecutive blocks is consistent. As shown
in Fig. \ref{fig:Blockchain Structure.}, blocks are connected through
the VDF result $y_{h}$. 

\subsection{Discussions on Security and Efficiency}

Remark that most VDF-based consensus protocols face the dilemma between
energy efficiency and security. For instance, R3V \cite{Raikwar2021}
restricts each miner to compute the VDF with a unique input. In this
way, R3V can prevent parallel computation and reduce energy consumption;
however, such design does not involve the block payload data or its
digest, i.e., the block hash, in the VDF input so that R3V cannot
protect the on-chain information from being tampered with. In contrast,
\cite{Long2019,Han2020,Mirkin2024} include the block hash as one
of the VDF inputs to protect the payload data, but miners can revise
the block payload data and create multiple different digests as the
VDF inputs by, e.g., simply exchanging the order of transactions,
for parallel VDF computing to obtain extra mining profits. It now
seems that the two goals, efficiency and security, cannot be achieved
simultaneously.

To address this dilemma, PoVD gives an alternative design. To guarantee
the blockchain integrity, PoVD involves the block payload data $m_{h,i}$,
represented by the hash value, into the VDF computation. But modifying
the payload data $m_{h,i}$ does not affect the time $T_{h,i}$ for
a miner to obtain the VDF output, since the verifiable delay $T_{h,i}$
solely depends on the previous VDF output $y_{h-1}$ and the miner
$\textrm{ID}_{i}$. Miners can generate multiple blocks by parallel
computing; however, all of these blocks require the time $T_{h,i}$
to be generated. That is, miners cannot accelerate block generation
by increasing computational power. 

Note that PoVD involves the miner's identity into the VDF computation,
just like most VDF-based consensus protocols. It is inevitable since
the mining process must associate with an identity. However, if a
miner can create multiple identities, i.e., launch a Sybil attack,
it can increase the chance to be the first one to generate a new block.
To address this, resource-based Sybil resistance mechanisms are typically
employed to make identity creation prohibitively costly \cite{Iqbal2021,Xiao2020survey}.
One typical solution is to introduce a registration process for newly-joined
miners. For example, in protocols such as Algorand \cite{Gilad2017}
and OmniLedger \cite{Kokoris2018}, newly-joined miners are required
to stake some funds and sign a registration transaction with their
private key. The miner is allowed to participate in the consensus
process only if the registration transaction is confirmed by the network.
Other existing miners can verify the identity of the block proposer
using the public key. In these solutions, creating multiple identities
incurs both time and monetary costs. These existing blockchain projects
based on Algorand and OmniLedger illustrate that the registration
process is an effective method against Sybil attacks.

Anyway, one miner can still generate more than one block at the same
height, to split the network and compromise the ledger consistency.
Specifically, a malicious miner may broadcast multiple blocks at the
same height to different peers, referred to as equivocation, to increase
the chance that one of them gets accepted into the main chain. This
issue can be addressed by using Casper, a stake-based finality mechanism
\cite{Buterin2017}. In Casper, only one block can be finalized, so
mining multiple blocks at the same height becomes unprofitable. Moreover,
Casper's voting process can identify a miner's misbehavior of broadcasting
multiple blocks at the same height, and such equivocation will be
detected and penalized via slashing. Meanwhile, a malicious miner
may still hold the generated blocks and release them strategically
to launch a long-range double-spending attack by maintaining a private
branch. This attack can be defended by the finalization mechanism
based on validator voting. In Casper, a checkpoint becomes finalized
after it has received sufficient validator support and its validity
has been further confirmed in the following stage. Therefore, once
the checkpoint is finalized, the malicious competing branch can no
longer replace it. Hence, even a long-range private fork revealed
later cannot revert finalized checkpoints, which greatly enhances
the difficulty of such manipulation.

Besides the above attacks, the attackers may launch more advanced
attacks such as the grinding attack. Miner $i$ may still mine multiple
candidate blocks with different payload data $m_{h,i}$ and find $y_{h,i}^{*}$
to minimize the block generation time for the next block $T_{h+1,i}$.
If miner $i$'s block is successfully included in the main chain with
such a carefully chosen $y_{h,i}^{*}$, then it may take advantage
of winning in the next height since it chooses a short $T_{h+1,i}$.
To mitigate this attack, we can use a binding commitment mechanism
\cite{Choi2023,Syta2017}, under which each miner must broadcast a
claim to bind the digest of the payload data before starting the mining
process, i.e., before calculating the VDF. When other miners receive
the block corresponding to the previously received claim, they firstly
validate that the block contains the same payload data as the claim,
so that grinding by changing the payload after mining is not applicable.
Besides, they perform a local timing check, which is to validate that
the duration between receiving the claim and the corresponding block
must be larger than $T_{h,i}$. This mechanism ensures that the claim
is generated before the VDF calculation begins. Otherwise, if a miner
attempts to launch a grinding attack and broadcasts the claim after
completing the VDF computation, it must wait for roughly another $T_{h,i}$,
which greatly reduces the benefit of such an attack. Only when both
validations pass will the verifier continue the subsequent process,
otherwise this block will be considered invalid locally.

\section{\label{sec:Protocol Modeling}Performance Analysis}

\subsection{Rethink the Mining Model}

As we have shown in \eqref{eq:T gener}, PoVD can flexibly adjust
the block time distribution by setting the function $g$ properly.
Note that most PoW-style consensus protocols can only set the average
block time by adjusting mining difficulty, while the block time distribution
always obeys a Bernoulli distribution. It can be viewed as an extra
advantage of PoVD beyond its communication- and computation-efficiency.
However, the mining process becomes more difficult to model under
an adjustable block time distribution. We have to rethink the mining
model and construct a more general mining process model with arbitrary
block time distributions before quantifying PoVD's improvement.

According to the PoVD protocol, each miner's time parameter $T_{h,i}$
is independent, identically distributed, and is determined by the
function $g$ in \eqref{eq:T gener}. Therefore, we drop the subscript
indicating the block height and miner identity and define the random
variable $T$ to represent the number of rounds a miner needs to generate
a new block. The probability mass function (PMF) of $T$ is denoted
as $\textrm{Pr}\left(T=r\right)$. Furthermore, we introduce $\tilde{p}\left(r\right)$
to represent the probability that a miner generates a block in round
$r$ after the previous block has been generated. Notably, $\tilde{p}\left(0\right)=0$.
The PMF of $T$ is expressed as 
\begin{equation}
\textrm{Pr}\left(T=r\right)=\tilde{p}\left(r\right)\prod_{k=0}^{r-1}\left(1-\tilde{p}\left(k\right)\right),\:\textrm{for}\:r=0,1,...\label{eq:T and pr}
\end{equation}
The meaning of \eqref{eq:T and pr} is clear. A miner successfully
mining a new block in round $r$ means that it does not generate any
block in the first $r-1$ rounds, but succeeds in round $r$. 

If we would like to let $T$ follow a geometric distribution like
PoW, we can set the function $g$ as follows:
\begin{align}
g\left(x\right) & =r,\:\textrm{if}\:\left(1-p\right)^{r}\leq\frac{x}{\left|\mathbb{C}\right|}<\left(1-p\right)^{r-1},\label{eq:deltaT-1-1}
\end{align}
where $0<p<1$, $x=0,1,...,\left|\mathbb{C}\right|$, and $r=1,2,...$
In \eqref{eq:deltaT-1-1}, as $r$ increases, the range of $x$ that
satisfies the corresponding inequalities decreases geometrically.
Since the input $x$ is uniformly distributed over $\mathbb{C}$ by
the $\textrm{Hash}$ function, the resulting distribution of $T$
obeys:\setcounter{equation}{14}
\begin{figure*}
\begin{equation}
\mathsf{F}\left(\tilde{p}\left(r\right)\right)=1-\sum_{r=1}^{\infty}\prod_{k=0}^{r-1}\prod_{j=1}^{d-1}n\tilde{p}\left(r\right)\left(1-\tilde{p}\left(r\right)\right)^{n-1}\left(1-\tilde{p}\left(k\right)\right)^{n}\left(1-\tilde{p}\left(r+j\right)\right)^{n\left(1-w_{j}\right)}.\label{eq:fork rate}
\end{equation}
\rule[0.5ex]{2.04\columnwidth}{0.5pt}
\end{figure*}
\setcounter{equation}{9}
\begin{align}
\textrm{Pr}\left(T=r\right) & =\left(1-p\right)^{r-1}-\left(1-p\right)^{r}\nonumber \\
 & =p\left(1-p\right)^{r-1},\:\textrm{for}\:r=1,2,...\label{eq:g-r}
\end{align}
One can easily obtain the block time distribution $\tilde{p}\left(r\right)=p$
for all $r=1,2,...$ In other words, the probability that a miner
mines a new block is identical for every round, which aligns with
the property of PoW. 

In the following analysis, we focus on the block time distribution
$\tilde{p}\left(r\right)$, since the random variable $T$ can be
easily characterized by $\tilde{p}\left(r\right)$. Let the random
variable $U_{r}$ be the number of blocks generated in round $r$
with no blocks generated in the previous $r-1$ rounds. Then, we have
\begin{align}
\textrm{Pr}\left(U_{r}=1\right)= & n\tilde{p}\left(r\right)\left(1-\tilde{p}\left(r\right)\right)^{n-1}\prod_{k=0}^{r-1}\left(1-\tilde{p}\left(k\right)\right)^{n},\label{eq:Ur=00003D1}\\
\textrm{Pr}\left(U_{r}\geq1\right)= & \left(1-\left(1-\tilde{p}\left(r\right)\right)^{n}\right)\prod_{k=0}^{r-1}\left(1-\tilde{p}\left(k\right)\right)^{n},\label{eq:Ur>=00003D1}
\end{align}
which represent the probability that exactly one block is generated
and the probability that at least one block is generated, respectively.
These two equations share a common term $\prod_{k=0}^{r-1}\left(1-\tilde{p}\left(k\right)\right)^{n}$,
which represents the probability that no block is generated in the
first $r-1$ rounds. Equations \eqref{eq:Ur=00003D1} and \eqref{eq:Ur>=00003D1}
are the basis of the blockchain performance analysis regarding fork
rate and block time under a given distribution $\tilde{p}\left(r\right)$.

\subsection{Fork Rate and Block Time}

We first derive the fork rate. Assume that the network is in a consistent
state in round 0, i.e., every miner reaches an agreement on the main
chain. We consider the event $A_{r}$, where only one block is generated
in round $r$, and no additional blocks are generated by any other
miners in the subsequent $d$ rounds during its propagation process.
The probability of the event $A_{r}$ is given by
\begin{equation}
\textrm{Pr}\left(A_{r}\right)=\textrm{Pr}\left(U_{r}=1\right)\prod_{j=1}^{d-1}\left(1-\tilde{p}\left(r+j\right)\right)^{n\left(1-w_{j}\right)}.\label{eq:PAR}
\end{equation}
We can use the event $A=\bigcup_{r=1}^{\infty}A_{r}$ to describe
that only one block is generated from a consistent state in round
0. Note that the complement of event $A$ means at least one fork
occurs, and thus, the fork rate $\mathsf{F}\left(\tilde{p}\left(r\right)\right)$
can be expressed as
\begin{equation}
\mathsf{F}\left(\tilde{p}\left(r\right)\right)=1-\textrm{Pr}\left(\bigcup_{r=1}^{\infty}A_{r}\right)\overset{\left(a\right)}{=}1-\sum_{r=1}^{\infty}\textrm{Pr}\left(A_{r}\right),\label{eq:fork rate simple}
\end{equation}
\setcounter{equation}{15}where (a) is because every $A_{r}$ is exclusive.
By substituting \eqref{eq:Ur=00003D1} and \eqref{eq:PAR} into \eqref{eq:fork rate simple},
we obtain the expression for the fork rate in \eqref{eq:fork rate}
with an arbitrary block time distribution $\tilde{p}\left(r\right)$.

According to \eqref{eq:fork rate}, the fork rate is composed of three
components. The term $\prod_{j=1}^{d-1}\left(1-\tilde{p}\left(r+j\right)\right)^{n\left(1-w_{j}\right)}$
represents the impact of network propagation, $\prod_{k=0}^{r-1}\left(1-\tilde{p}\left(k\right)\right)^{n}$
is the probability that no block is generated in the first $r-1$
rounds, and $n\tilde{p}\left(r\right)\left(1-\tilde{p}\left(r\right)\right)^{n-1}$
corresponds to the event where only one miner generates a block in
round $r$. The combination of these three parts indicates that the
network has no fork. 

The fork rate essentially reflects the consistency of the blockchain.
Forks mean that different miners have different visions on the main
chain. Therefore, the fork rate \eqref{eq:fork rate} is an important
quantitative metric of blockchain's consistency describing the fundamental
property of distributed systems. The lower the fork rate, the more
consistent the blockchain is. 

Block time is the expected number of rounds between two consecutive
blocks. The expected number of rounds for block generation can be
obtained from \eqref{eq:Ur>=00003D1}, yielding the (average) block
time $\mathsf{B}\left(\tilde{p}\left(r\right)\right)$:
\begin{align}
\mathsf{B}\left(\tilde{p}\left(r\right)\right) & =\sum_{r=1}^{\infty}r\textrm{Pr}\left(U_{r}\geq1\right)\nonumber \\
 & =\sum_{r=1}^{\infty}\prod_{k=0}^{r-1}r\left(1-\left(1-\tilde{p}\left(r\right)\right)^{n}\right)\left(1-\tilde{p}\left(k\right)\right)^{n}.\label{eq:average blocktime}
\end{align}
Notably, our derivation is based on the probability that at least
one block is generated, i.e., $\textrm{Pr}\left(U_{r}\geq1\right)$.
Hence, the inverse of block time accurately reflects the main chain's
growth rate \cite{Fujihara2024}. The smaller the block time, the
more transactions the blockchain can process per unit of time, resulting
in higher throughput. Therefore, block time in \eqref{eq:average blocktime}
can be viewed as a critical metric for quantifying the throughput
of a blockchain. 

In conclusion, we rebuild a blockchain mining model with an arbitrary
block time distribution and derive the fork rate in \eqref{eq:fork rate}
and block time in \eqref{eq:average blocktime} in closed forms. Note
that the above modeling and analysis are derived from a consistent
state. The obtained fork rate and block time are accurate under this
assumption and provide good approximations even without the assumption,
as shown in the experimental results. 

\subsection{PoW as an Example}

Note that our mining model also works for PoW. Given the mining difficulty
and the number of hash trials in a round, the probability that a miner
generates a block in a round in PoW is constant. Hence, $T$ can be
modeled as a geometric distribution, and the block time distribution
of PoW can be denoted as $\tilde{p}_{a}\left(r\right)=p$, where $r=1,2,...$
By substituting $\tilde{p}_{a}\left(r\right)=p$ and \eqref{eq:network value}
into \eqref{eq:fork rate}, we can derive the fork rate $\mathsf{F}\left(\tilde{p}_{a}\left(r\right)\right)$
for PoW as follows:
\begin{align}
\mathsf{F}\left(\tilde{p}_{a}\left(r\right)\right) & =1-\frac{np\left(1-p\right)^{n\left(d-\sum_{i=0}^{d-1}w_{i}\right)}}{1-\left(1-p\right)^{n}}\nonumber \\
 & =1-\frac{np\left(1-p\right)^{nd\left(1-c\right)}}{1-\left(1-p\right)^{n}}.\label{eq:pow forkrate}
\end{align}
In fact, \eqref{eq:pow forkrate} is a conditional probability. The
denominator, $1-\left(1-p\right)^{n}$, represents the probability
that blocks are generated in a given round. The term $np\left(1-p\right)^{nd\left(1-c\right)}$
represents the probability that only one block is generated, and the
exponent $nd\left(1-c\right)$ represents the number of miners who
have not received the block during the $d$ rounds propagation process.
In \cite{Decker2013}, the authors obtained similar results by calculating
the probability of no block being generated during propagation.

Similarly, by substituting $\tilde{p}_{a}\left(r\right)=p$ into \eqref{eq:average blocktime},
we can derive the average block time of PoW $\mathsf{B}\left(\tilde{p}_{a}\left(r\right)\right)$
as follows:
\begin{align}
\mathsf{B}\left(\tilde{p}_{a}\left(r\right)\right)= & \sum_{r=1}^{\infty}r\left(1-\left(1-p\right)^{n}\right)\left(1-p\right)^{n\left(r-1\right)}\nonumber \\
= & \frac{1}{1-\left(1-p\right)^{n}}.\label{eq:PoW average blocktime}
\end{align}
Note that, in PoW, the probability that at least one block is generated
within a round is $1-\left(1-p\right)^{n}$. As a result, the time
to generate new blocks can also be modeled as a geometric distribution,
with an expectation of $\frac{1}{1-\left(1-p\right)^{n}}$, which
is exactly the block time of PoW. As shown in \eqref{eq:PoW average blocktime},
increasing either $p$ or $n$ reduces the block time. This is because
a higher $p$ or $n$ increases the probability that miners mine a
new block in each round, which reduces the number of rounds to generate
a new block and thus lowers the block time. As one can see, both $\mathsf{F}\left(\tilde{p}_{a}\left(r\right)\right)$
and $\mathsf{B}\left(\tilde{p}_{a}\left(r\right)\right)$ of PoW are
determined by the parameters of $p$ and $n$. Meanwhile, PoVD can
flexibly adjust the fork rate and block time by setting the block
time distribution $\tilde{p}\left(r\right)$ via the function $g$,
which is another significant advantage of PoVD.

\section{\label{sec:Protocol Analysis}PoVD Block Time Re-Distribution}

\subsection{$\delta$-spaced Block Time Distribution}

The fork rate and block time are two key metrics of blockchains, serving
as quantitative indicators of consistency and throughput, respectively.
A lower fork rate reflects higher consistency, while a shorter block
time indicates higher throughput. However, although reducing block
time can improve throughput, it also increases the probability that
more than one block is generated during the propagation of another
block, leading to a higher fork rate. Therefore, we should balance
the trade-off between fork rate and block time in the blockchain design.

Therefore, can we reduce the fork rate $\mathsf{F}\left(\tilde{p}\left(r\right)\right)$
with the same block time $\mathsf{B}\left(\tilde{p}\left(r\right)\right)$,
by adjusting the distribution $\tilde{p}\left(r\right)$? Due to the
complicated expressions of $\mathsf{F}\left(\tilde{p}\left(r\right)\right)$
and $\mathsf{B}\left(\tilde{p}\left(r\right)\right)$, this is a very
challenging problem. So far, the optimal block time distribution to
minimize the fork rate under the block time constraint remains open.
However, we would like to provide the $\delta$-spaced block time
distribution, denoted by $\tilde{p}_{\delta}\left(r\right)$, that
can achieve a lower fork rate than PoW. The $\delta$-spaced block
time distribution $\tilde{p}_{\delta}\left(r\right)$ is given by
\begin{equation}
\tilde{p}_{\delta}\left(r\right)=\begin{cases}
p_{\delta}, & r=1,\delta+1,2\delta+1,...\\
0, & \mathrm{otherwise},
\end{cases}\label{eq:expression of p}
\end{equation}
where $\delta$ is a positive integer. Essentially, under the $\delta$-spaced
block time distribution $\tilde{p}_{\delta}\left(r\right)$, a miner
can only mine a block every $\delta$ rounds. In this case, we can
set the function $g$ to be
\begin{align}
g\left(x\right) & =r,\:\textrm{if}\:\left(1-p_{\delta}\right)^{\frac{r+\delta-1}{\delta}}\leq\frac{x}{\left|\mathbb{C}\right|}<\left(1-p_{\delta}\right)^{\frac{r-1}{\delta}},\label{eq:deltaT-1}
\end{align}
where $r=1,\delta+1,2\delta+1,...,$ and $0\leq x\leq\left|\mathbb{C}\right|$.

In a network with the maximum propagation delay $d$, the $\delta$-spaced
block time distribution $\tilde{p}_{\delta}\left(r\right)$ with $\delta=d$
yields the fork rate $\mathsf{F}\left(\tilde{p}_{\delta}\left(r\right)\right)$
and the block time $\mathsf{B}\left(\tilde{p}_{\delta}\left(r\right)\right)$,
given by
\begin{align}
\mathsf{F}\left(\tilde{p}_{\delta}\left(r\right)\right)= & 1-\frac{np_{\delta}\left(1-p_{\delta}\right)^{n-1}}{1-\left(1-p_{\delta}\right)^{n}},\label{eq:VDF forkrate}\\
\mathsf{B}\left(\tilde{p}_{\delta}\left(r\right)\right)= & \frac{\delta\left(1-p_{\delta}\right)^{n}}{1-\left(1-p_{\delta}\right)^{n}}+1.\label{eq:VDF average blocktime}
\end{align}

\subsection{Superiority of PoVD}

To further compare PoVD with PoW under the same average block time,
we give two lemmas that will be used to bound the threshold of the
network capability indicator.
\begin{lem}
\label{lem:max n}Let $q\in\left(0,1\right)$ and $d=1,2...,$ be
constants. Then, for any $n=1,2,...$, the following inequality holds:
\begin{equation}
\frac{\ln\omega\left(n\right)/d}{\ln\left(1-q\right)}-\frac{1}{n}<\frac{\ln\left(\frac{1}{d}-\frac{\ln\left(1+dq-q\right)}{d\ln\left(1-q\right)}\right)}{\ln\left(1-q\right)}.\label{eq:max n}
\end{equation}
\end{lem}
\begin{proof}
We first prove $\omega^{\prime}\left(n\right)<0$. Let $f_{1}\left(n\right)=\left(\frac{1+dq-q}{1-q}\right)^{\frac{1}{n}}-1$,
which is the numerator of $\omega\left(n\right)$, and $f_{2}\left(n\right)=\left(\frac{1}{1-q}\right)^{\frac{1}{n}}-1$,
which is the denominator. Note that for any $q\in\left(0,1\right)$,
we have $\frac{1+dq-q}{1-q}>\frac{1}{1-q}>1$, and we further have
$f_{1}\left(n\right)>f_{2}\left(n\right)>0$ and $f_{1}^{\prime}\left(n\right)<f_{2}^{\prime}\left(n\right)<0$.
Hence, the derivative of $\omega\left(n\right)$, i.e., $\omega^{\prime}\left(n\right)$,
is
\begin{align}
\omega^{\prime}\left(n\right) & =\frac{1}{\left(f_{2}\left(n\right)\right)^{2}}\left(f_{2}\left(n\right)f_{1}^{\prime}\left(n\right)-f_{1}\left(n\right)f_{2}^{\prime}\left(n\right)\right)\nonumber \\
 & <\frac{1}{\left(f_{2}\left(n\right)\right)^{2}}\left(f_{1}\left(n\right)f_{1}^{\prime}\left(n\right)-f_{1}\left(n\right)f_{2}^{\prime}\left(n\right)\right)\nonumber \\
 & =\frac{f_{1}\left(n\right)}{\left(f_{2}\left(n\right)\right)^{2}}\left(f_{1}^{\prime}\left(n\right)-f_{2}^{\prime}\left(n\right)\right)<0.
\end{align}
Therefore, the left side of \eqref{eq:max n} increases in $n$ since
\begin{equation}
\frac{\partial}{\partial n}\left(\frac{\ln\omega\left(n\right)/d}{\ln\left(1-q\right)}-\frac{1}{n}\right)=\frac{\omega^{\prime}\left(n\right)}{\omega\left(n\right)\ln\left(1-q\right)}+\frac{1}{n^{2}}>0,
\end{equation}
where $\omega\left(n\right)>0$ and $\ln\left(1-q\right)<0$. Finally,
we can prove \eqref{eq:max n}:
\begin{align}
 & \frac{\ln\omega\left(n\right)/d}{\ln\left(1-q\right)}-\frac{1}{n}<\underset{n\rightarrow\infty}{\lim}\left(\frac{\ln\omega\left(n\right)/d}{\ln\left(1-q\right)}-\frac{1}{n}\right)\nonumber \\
 & \overset{\left(a\right)}{=}\frac{1}{\ln\left(1-q\right)}\ln\frac{1}{d}\left(\underset{n\rightarrow\infty}{\lim}\frac{\exp\left(\frac{1}{n}\ln\left(\frac{1+dq-q}{1-q}\right)\right)-1}{\exp\left(\frac{1}{n}\ln\left(\frac{1}{1-q}\right)\right)-1}\right)\nonumber \\
 & \overset{\left(b\right)}{=}\frac{1}{\ln\left(1-q\right)}\ln\frac{1}{d}\left(\underset{n\rightarrow\infty}{\lim}\frac{\frac{1}{n}\ln\left(\frac{1+dq-q}{1-q}\right)}{\frac{1}{n}\ln\left(\frac{1}{1-q}\right)}\right)\nonumber \\
 & \overset{\left(c\right)}{=}\frac{1}{\ln\left(1-q\right)}\ln\left(\frac{1}{d}-\frac{\ln\left(1+dq-q\right)}{d\ln\left(1-q\right)}\right),
\end{align}
where (a) is because $\lim_{n\rightarrow\infty}\frac{1}{n}=0$ and
rewrites the formula, (b) is based on $\exp\left(x\right)-1\thicksim x$
as $x\rightarrow0$, and (c) simplifies the formula. 
\end{proof}
\begin{lem}
\label{lem:hu}Let $d=1,2...,$ be a constant. Then, for any $q\in\left(0,1\right)$,
the following holds:
\begin{equation}
\frac{\ln\left(\frac{1}{d}-\frac{\ln\left(1+dq-q\right)}{d\ln\left(1-q\right)}\right)}{\ln\left(1-q\right)}<\frac{d}{2}-\frac{1}{2}.\label{eq:hu}
\end{equation}
\end{lem}
\begin{proof}
We define the left-hand side of \eqref{eq:hu} as function $h\left(q\right)$:
\begin{align}
h\left(q\right) & =\frac{\ln\frac{1}{d}\left(1-\frac{\ln\left(1+dq-q\right)}{\ln\left(1-q\right)}\right)}{\ln\left(1-q\right)}\nonumber \\
 & =\frac{1}{h_{2}\left(q\right)}\ln\left(\frac{1}{d}-\frac{h_{1}\left(q\right)}{dh_{2}\left(q\right)}\right),
\end{align}
where $h_{1}\left(q\right)=\ln\left(1+dq-q\right),$ and $h_{2}\left(q\right)=\ln\left(1-q\right)$.
Next, we prove that $h\left(q\right)$ decreases in $q$. The derivative
of $h\left(q\right)$ is given by:
\begin{align}
h^{\prime}\left(q\right)= & \frac{1}{\left(h_{2}\left(q\right)-h_{1}\left(q\right)\right)\left(h_{2}\left(q\right)\right)^{2}}\nonumber \\
 & \cdot\left(\left(h_{1}\left(q\right)-h_{2}\left(q\right)\right)h_{2}^{\prime}\left(q\right)\ln\left(\frac{1}{d}-\frac{h_{1}\left(q\right)}{dh_{2}\left(q\right)}\right)\right.\nonumber \\
 & \left.-\left(\frac{h_{1}\left(q\right)}{h_{2}\left(q\right)}\right)^{\prime}\left(h_{2}\left(q\right)\right)^{3}\right).
\end{align}
The denominator of $h^{\prime}\left(q\right)$ is negative since $h_{2}\left(q\right)-h_{1}\left(q\right)=\ln\left(\left(1-q\right)/\left(1+\left(d-1\right)q\right)\right)<\ln1=0$.
We would like to show that the numerator of $h^{\prime}\left(q\right)$
is positive. Note that $h_{1}\left(q\right)$ increases in $q$ (since
$h_{1}^{\prime}\left(q\right)=\frac{d-1}{1+dq-q}>0$) and $h_{2}\left(q\right)$
decreases in $q$ (since $h_{2}^{\prime}\left(q\right)=\frac{1}{q-1}<0$),
so $h_{1}\left(q\right)/h_{2}\left(q\right)$ increases in $q$. Therefore,
we further have $\left(h_{1}\left(q\right)/h_{2}\left(q\right)\right)^{\prime}>0,$
and
\begin{align}
\frac{h_{1}\left(q\right)}{h_{2}\left(q\right)} & >\underset{q\rightarrow0^{+}}{\lim}\frac{h_{1}\left(q\right)}{h_{2}\left(q\right)}=\underset{q\rightarrow0^{+}}{\lim}\frac{\ln\left(1+\left(d-1\right)q\right)}{\ln\left(1-q\right)}\nonumber \\
 & \overset{\left(a\right)}{=}\underset{q\rightarrow0^{+}}{\lim}\frac{\left(d-1\right)q}{-q}=1-d,\label{eq:h1/h2=00003D1}
\end{align}
where (a) is based on $\ln\left(1+x\right)\thicksim x$ as $x\rightarrow0^{+}$.
Hence, we have $\ln\left(\frac{1}{d}-\frac{h_{1}\left(q\right)}{dh_{2}\left(q\right)}\right)<\ln\left(\frac{1}{d}\cdot d\right)=0$.
Now, we can summarize each term of the numerator of $h^{\prime}\left(q\right)$:
\begin{equation}
\begin{cases}
h_{1}\left(q\right)-h_{2}\left(q\right)=\ln\left(\frac{1+\left(d-1\right)q}{1-q}\right)>0,\\
h_{2}^{\prime}\left(q\right)=\frac{1}{q-1}<0,\\
\ln\left(\frac{1}{d}-\frac{h_{1}\left(q\right)}{dh_{2}\left(q\right)}\right)<0,\\
\left(\frac{h_{1}\left(q\right)}{h_{2}\left(q\right)}\right)^{\prime}>0,\\
\left(h_{2}\left(q\right)\right)^{3}=\left(\ln\left(1-q\right)\right)^{3}<0.
\end{cases}
\end{equation}
Therefore, the numerator of $h^{\prime}\left(q\right)$ is positive,
and $h\left(q\right)$ decreases in $q$. Finally, we can prove \eqref{eq:hu}:
\begin{align}
 & h\left(q\right)<\underset{q\rightarrow0^{+}}{\lim}h\left(q\right)\nonumber \\
 & \;\overset{\left(a\right)}{=}\underset{q\rightarrow0^{+}}{\lim}\frac{\ln\left(1+\left(\frac{1}{d}-\frac{\ln\left(1+dq-q\right)}{d\ln\left(1-q\right)}\right)-1\right)}{\ln\left(1-q\right)}\nonumber \\
 & \;\overset{\left(b\right)}{=}\underset{q\rightarrow0^{+}}{\lim}\frac{\frac{1}{d}-\frac{\ln\left(1+dq-q\right)}{d\ln\left(1-q\right)}-1}{-q}\nonumber \\
 & \;\overset{\left(c\right)}{=}\frac{1}{d}\underset{q\rightarrow0^{+}}{\lim}\frac{\left(1-d\right)\ln\left(1-q\right)-\ln\left(1+dq-q\right)}{-q\ln\left(1-q\right)}\nonumber \\
 & \;\overset{\left(d\right)}{=}\frac{1}{d}\underset{q\rightarrow0^{+}}{\lim}\frac{\left(1-d\right)\ln\left(1-q\right)-\ln\left(1+dq-q\right)}{q^{2}}\nonumber \\
 & \;\overset{\left(e\right)}{=}\frac{1}{d}\underset{q\rightarrow0^{+}}{\lim}\frac{\left(d-1\right)\frac{1}{1-q}-\frac{d-1}{1+dq-q}}{2q}\nonumber \\
 & \;\overset{\left(f\right)}{=}\frac{1}{d}\underset{q\rightarrow0^{+}}{\lim}\frac{d^{2}q-dq}{2q\left(1-q\right)\left(1+dq-q\right)}\nonumber \\
 & \;\overset{\left(g\right)}{=}\frac{1}{d}\underset{q\rightarrow0^{+}}{\lim}\frac{d^{2}q-dq}{2q}=\frac{d}{2}-\frac{1}{2},
\end{align}
where (a) introduces infinitesimal term, i.e.,
\begin{equation}
\underset{q\rightarrow0^{+}}{\lim}\left(\frac{1}{d}-\frac{\ln\left(1+dq-q\right)}{d\ln\left(1-q\right)}\right)-1=0,
\end{equation}
according to \eqref{eq:h1/h2=00003D1}, (b) substitutes the numerator
based on $\ln\left(1+x\right)\thicksim x$ as $x\rightarrow0$, (c)
simplifies the formula, (d) substitutes the infinitesimal term $\ln\left(1-q\right)\thicksim-q$,
(e) applies the L'Hopital's rule, (f) simplifies the formula, and
(g) calculates the nonzero limit in the denominator, i.e., $\underset{q\rightarrow0^{+}}{\lim}\left(1-q\right)=\underset{q\rightarrow0^{+}}{\lim}\left(1+dq-q\right)=1$.
\end{proof}
With the two lemmas above, we can summarize the superiority of PoVD
in terms of fork rate in the following theorem.
\begin{thm}
\label{thm:basic t1}Assume that a network with the maximum delay
$d$ has the network capability indicator satisfying $c\leq\frac{1}{2}-\frac{1}{2d}$.
By setting $\delta=d$, PoVD with the $\delta$-spaced block time
distribution $\tilde{p}_{\delta}\left(r\right)$ can achieve a lower
fork rate than PoW under the same average block time, i.e., 
\begin{align*}
\mathsf{F}\left(\tilde{p}_{\delta}\left(r\right)\right) & <\mathsf{F}\left(\tilde{p}_{a}\left(r\right)\right),\\
\mathsf{B}\left(\tilde{p}_{a}\left(r\right)\right) & =\mathsf{B}\left(\tilde{p}_{\delta}\left(r\right)\right).
\end{align*}
\end{thm}
\begin{proof}
According to \eqref{eq:pow forkrate} and \eqref{eq:VDF forkrate},
$\mathsf{F}\left(\tilde{p}_{\delta}\left(r\right)\right)<\mathsf{F}\left(\tilde{p}_{a}\left(r\right)\right)$
with $\delta=d$ is equivalent to:
\begin{equation}
\frac{p_{\delta}\left(1-p_{\delta}\right)^{n-1}}{1-\left(1-p_{\delta}\right)^{n}}>\frac{p\left(1-p\right)^{nd\left(1-c\right)}}{1-\left(1-p\right)^{n}}.\label{eq:eq_th1}
\end{equation}
Substituting \eqref{eq:PoW average blocktime} and \eqref{eq:VDF average blocktime}
into $\mathsf{B}\left(\tilde{p}_{\delta}\left(r\right)\right)=\mathsf{B}\left(\tilde{p}_{a}\left(r\right)\right)$
with $\delta=d$ yields
\begin{equation}
\left(1-p_{\delta}\right)^{n}=\frac{\left(1-p\right)^{n}}{d\left(1-\left(1-p\right)^{n}\right)+\left(1-p\right)^{n}}.\label{eq:relation pow vdf}
\end{equation}
Combining \eqref{eq:eq_th1} and \eqref{eq:relation pow vdf}, and
letting $q=1-\left(1-p\right)^{n}\in\left(0,1\right)$, we derive
the condition for the network capability indicator $c$ that satisfies
both $\mathsf{F}\left(\tilde{p}_{\delta}\left(r\right)\right)<\mathsf{F}\left(\tilde{p}_{a}\left(r\right)\right)$
and $\mathsf{B}\left(\tilde{p}_{a}\left(r\right)\right)=\mathsf{B}\left(\tilde{p}_{\delta}\left(r\right)\right)$:
\begin{equation}
c<c^{\mathrm{th}}\triangleq1-\frac{1}{d}\left(1+\frac{\ln\omega\left(n\right)/d}{\ln\left(1-q\right)}-\frac{1}{n}\right),\label{eq:cub}
\end{equation}
where $\omega\left(n\right)$ is given by
\begin{equation}
\omega\left(n\right)=\frac{\left(\frac{1+dq-q}{1-q}\right)^{\frac{1}{n}}-1}{\left(\frac{1}{1-q}\right)^{\frac{1}{n}}-1}.
\end{equation}
Next, we derive the lower bound for $c^{\mathrm{th}}$:
\begin{align}
c^{\mathrm{th}} & \overset{\left(a\right)}{>}1-\frac{1}{d}\left(1+\frac{\ln\left(\frac{1}{d}-\frac{\ln\left(1+dq-q\right)}{d\ln\left(1-q\right)}\right)}{\ln\left(1-q\right)}\right)\nonumber \\
 & \overset{\left(b\right)}{>}1-\frac{1}{d}\left(1+\frac{d}{2}-\frac{1}{2}\right)=\frac{1}{2}-\frac{1}{2d},
\end{align}
where (a) is because of Lemma \ref{lem:max n}, and (b) is because
of Lemma \ref{lem:hu}. Therefore, given the network condition $c\leq\frac{1}{2}-\frac{1}{2d}$,
when $\mathsf{B}\left(\tilde{p}_{a}\left(r\right)\right)=\mathsf{B}\left(\tilde{p}_{\delta}\left(r\right)\right)$,
the inequality $\mathsf{F}\left(\tilde{p}_{\delta}\left(r\right)\right)<\mathsf{F}\left(\tilde{p}_{a}\left(r\right)\right)$
always holds.
\begin{figure}[t]
\centering{}\includegraphics[width=0.45\textwidth]{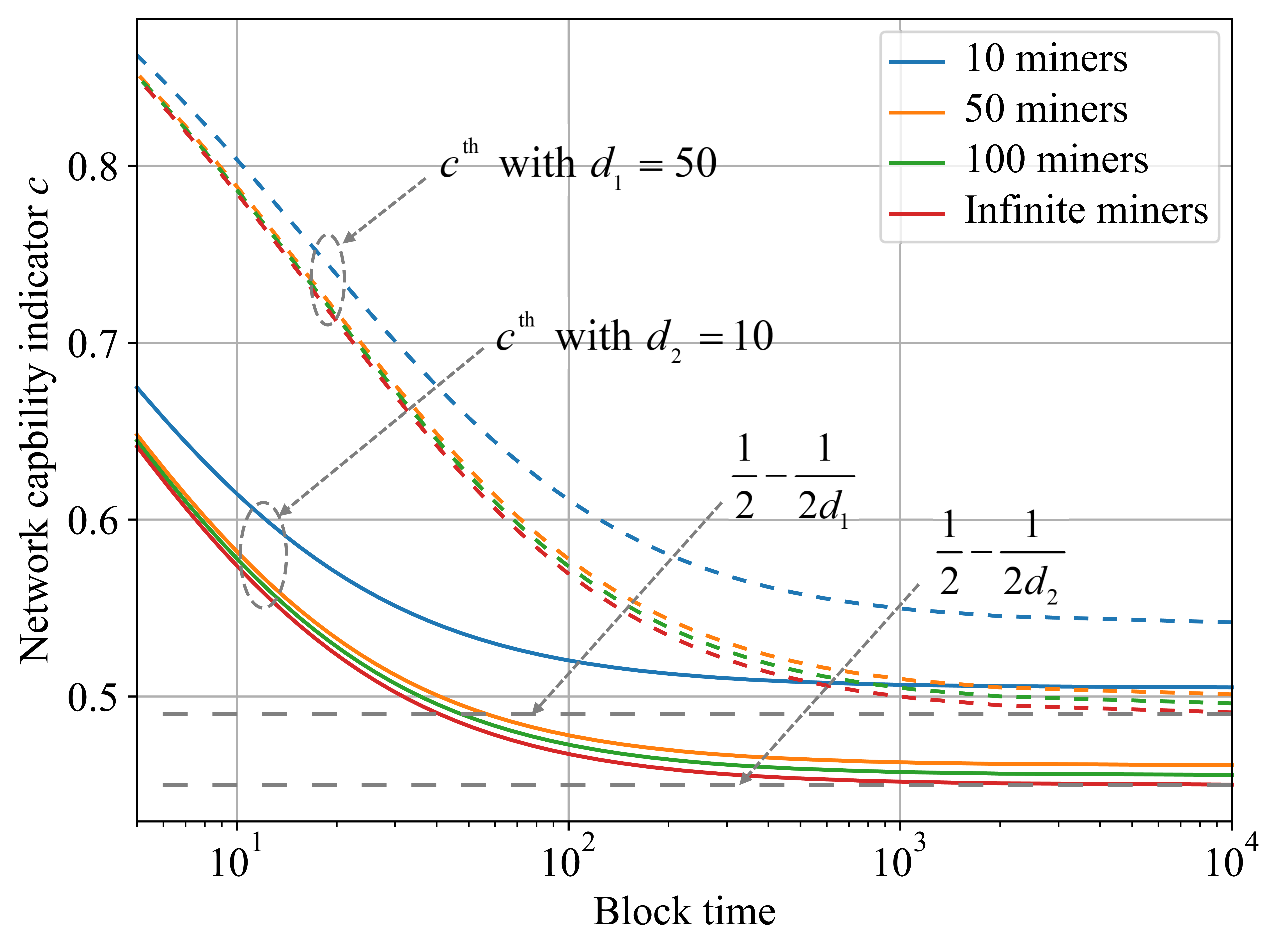}\caption{Threshold $c^{\mathrm{th}}$ of network capability indicator under
different block time.\label{fig:Boundary network}}
\vspace{-0.4cm}
\end{figure}
\end{proof}
Theorem \ref{thm:basic t1} indicates that, as long as the network
capability indicator satisfies $c\leq\frac{1}{2}-\frac{1}{2d}$, PoVD
can achieve a lower fork rate than PoW with the same block time by
using the $\delta$-spaced block time distribution $\tilde{p}_{\delta}\left(r\right)$.
Note that Theorem \ref{thm:basic t1} is a sufficient criterion for
PoVD to achieve higher blockchain consistency than PoW without sacrificing
throughput, since Theorem \ref{thm:basic t1} is based on the $\delta$-spaced
block time distribution $\tilde{p}_{\delta}\left(r\right)$, rather
than the optimal distribution. This further highlights the great potential
of PoVD. 

Note that some networks with long-tail latency do not satisfy the
condition of $c\leq\frac{1}{2}-\frac{1}{2d}$. These networks may
just need much fewer rounds than $d$ to distribute the block among
most miners, e.g., 95\% or more; however, it may take a very lengthy
time to ensure that all the miners receive the message. In this case,
we can use a smaller $\delta<d$, and PoVD still outperforms PoW in
terms of fork rate. It is verified by our experiments in Section \ref{sec:Simulation}.
The reason behind choosing $\delta<d$ is that propagation can be
approximately considered complete after $\delta$ rounds. This allows
the network capability indicator with delay $\delta$ to satisfy the
condition $c\leq\frac{1}{2}-\frac{1}{2\delta}$. Moreover, from a
practical perspective, since the majority of miners have already reached
a consensus on the main chain after $\delta$ rounds of propagation,
the few remaining miners who have not yet received the block are unlikely
to fork.

Moreover, the proof of Theorem \ref{thm:basic t1} provides $c^{\mathrm{th}}$
in \eqref{eq:cub} as the accurate threshold of network capability
indicator beyond the sufficient condition $c\leq\frac{1}{2}-\frac{1}{2d}$
in the statement of Theorem \ref{thm:basic t1}. PoVD can outperform
PoW, as long as $c\leq c^{\mathrm{th}}$ holds, and meanwhile $c\leq\frac{1}{2}-\frac{1}{2d}$
is just a stricter condition on the network. (Fig. \ref{fig:Boundary network}
shows that $c^{\mathrm{th}}>\frac{1}{2}-\frac{1}{2d}$ always holds.)
As $\mathsf{B}\left(\tilde{p}_{a}\left(r\right)\right)\rightarrow\infty$
and $n\rightarrow\infty$, the threshold $c^{\mathrm{th}}$ gradually
approaches the lower bound $\frac{1}{2}-\frac{1}{2d}$. The figure
also shows that when the block time is smaller, i.e., a high blockchain
throughput is required, PoVD imposes weaker network requirements (and
thus a wider range of $c$), which implies that PoVD offers greater
design space for improving consistency under some high throughput
requirements.

\cprotect\section{\label{sec:Simulation}Experimental Results
\begin{figure}[t]
\protect\centering{}\protect\includegraphics[width=0.45\textwidth]{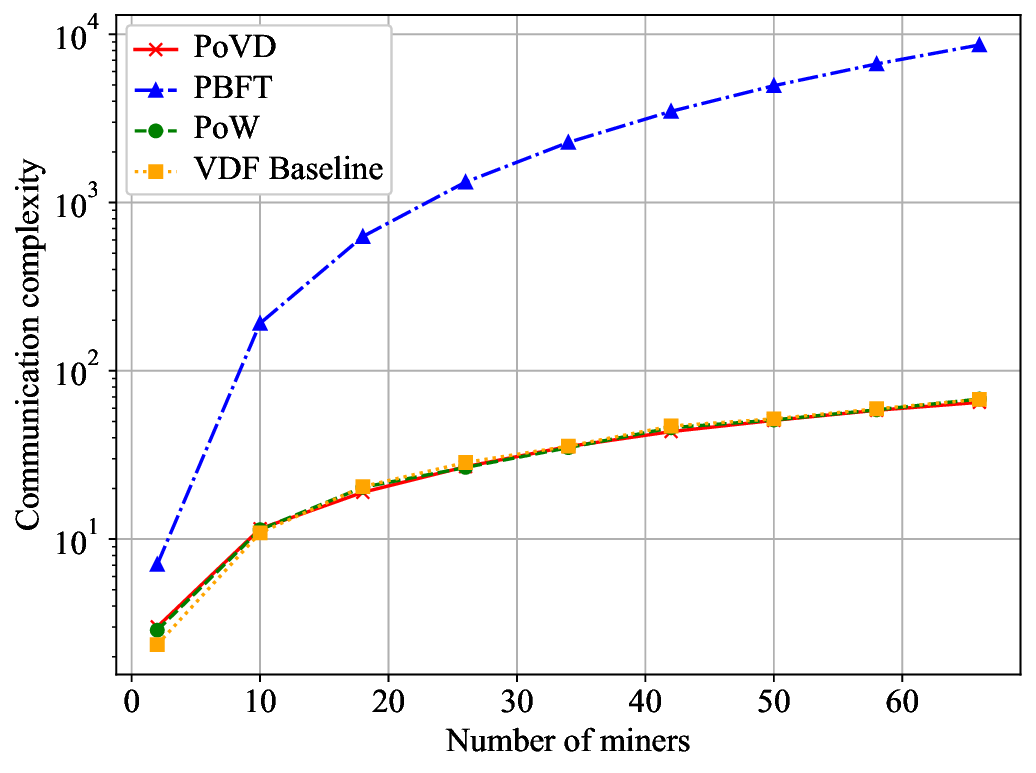}\protect\caption{Communication complexity of different consensus protocols under varying
numbers of miners. (Communication complexity is defined as the average
number of messages exchanged to generate a new block.)\label{fig:Comunication Complexity}}
\vspace{-0.4cm}\protect
\end{figure}
\begin{figure}[t]
\protect\centering{}\protect\includegraphics[width=0.45\textwidth]{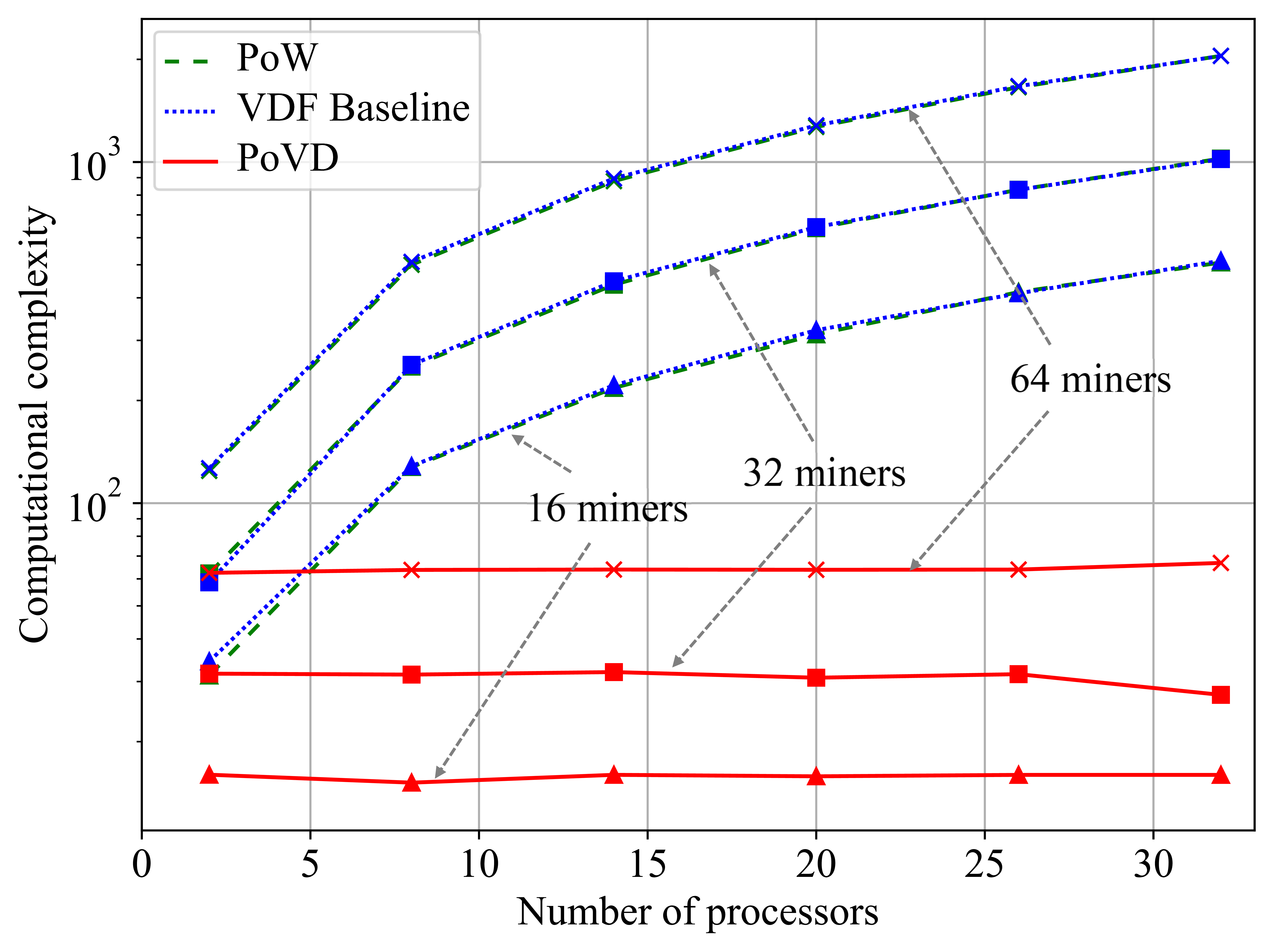}\protect\caption{Computational complexity of different consensus protocols with different
numbers of processors. (Computational complexity is defined as the
average number of oracles performed by a miner to generate a block.)\label{fig:Computation Complexity}}
\vspace{-0.4cm}\protect
\end{figure}
\begin{figure}[t]
\protect\centering{}\protect\includegraphics[width=0.45\textwidth]{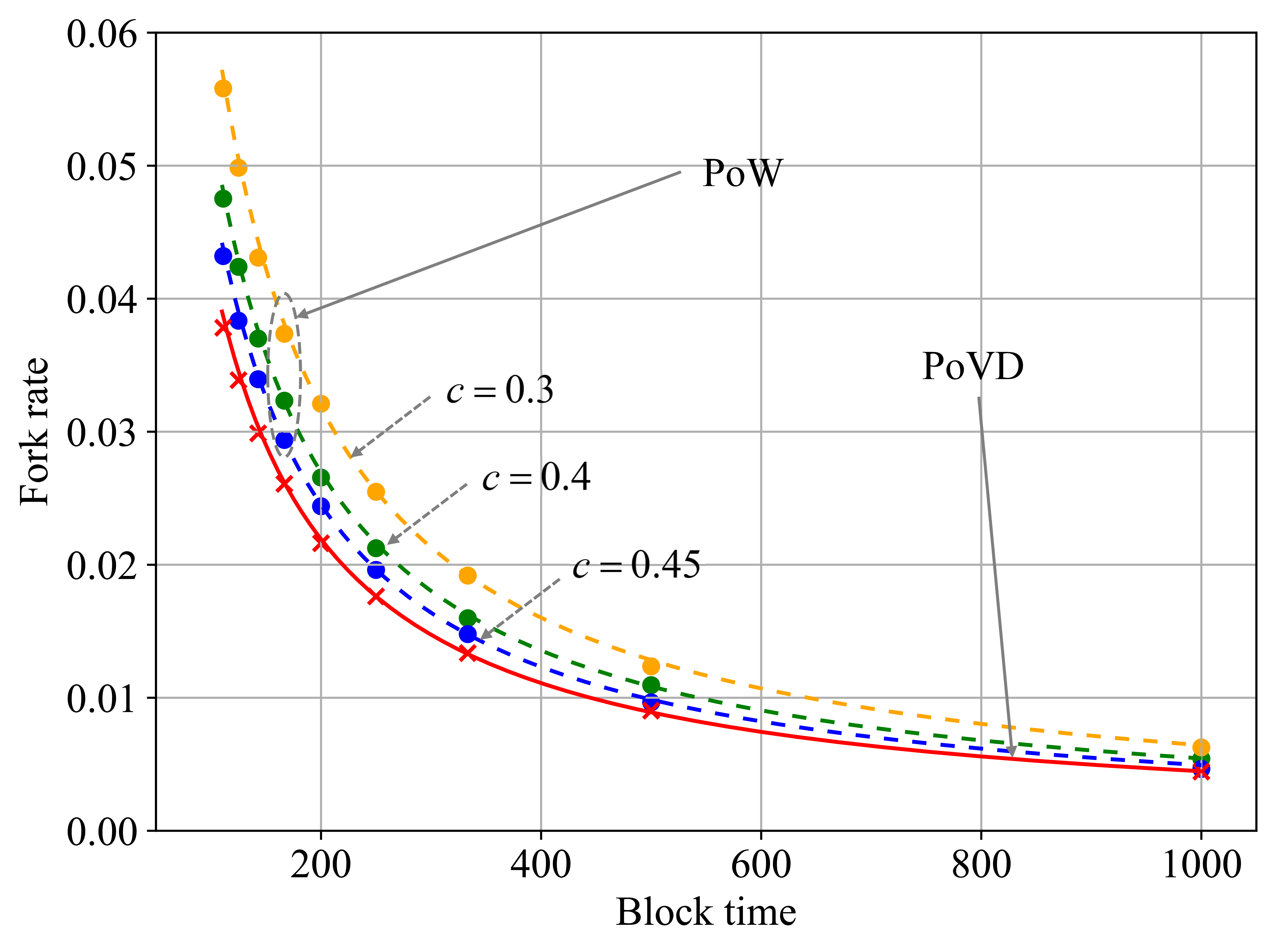}\protect\caption{Fork rate of PoW and PoVD under different block time with $n=10$,
and $d=\delta=10$. (The curves are analytical results, and the markers
are experimental ones.) \label{fig:n=00003D10, c=00003D0.4 and l=00003D10}}
\vspace{-0.4cm}\protect
\end{figure}
\begin{figure}
\protect\begin{centering}
\protect\includegraphics[width=0.45\textwidth]{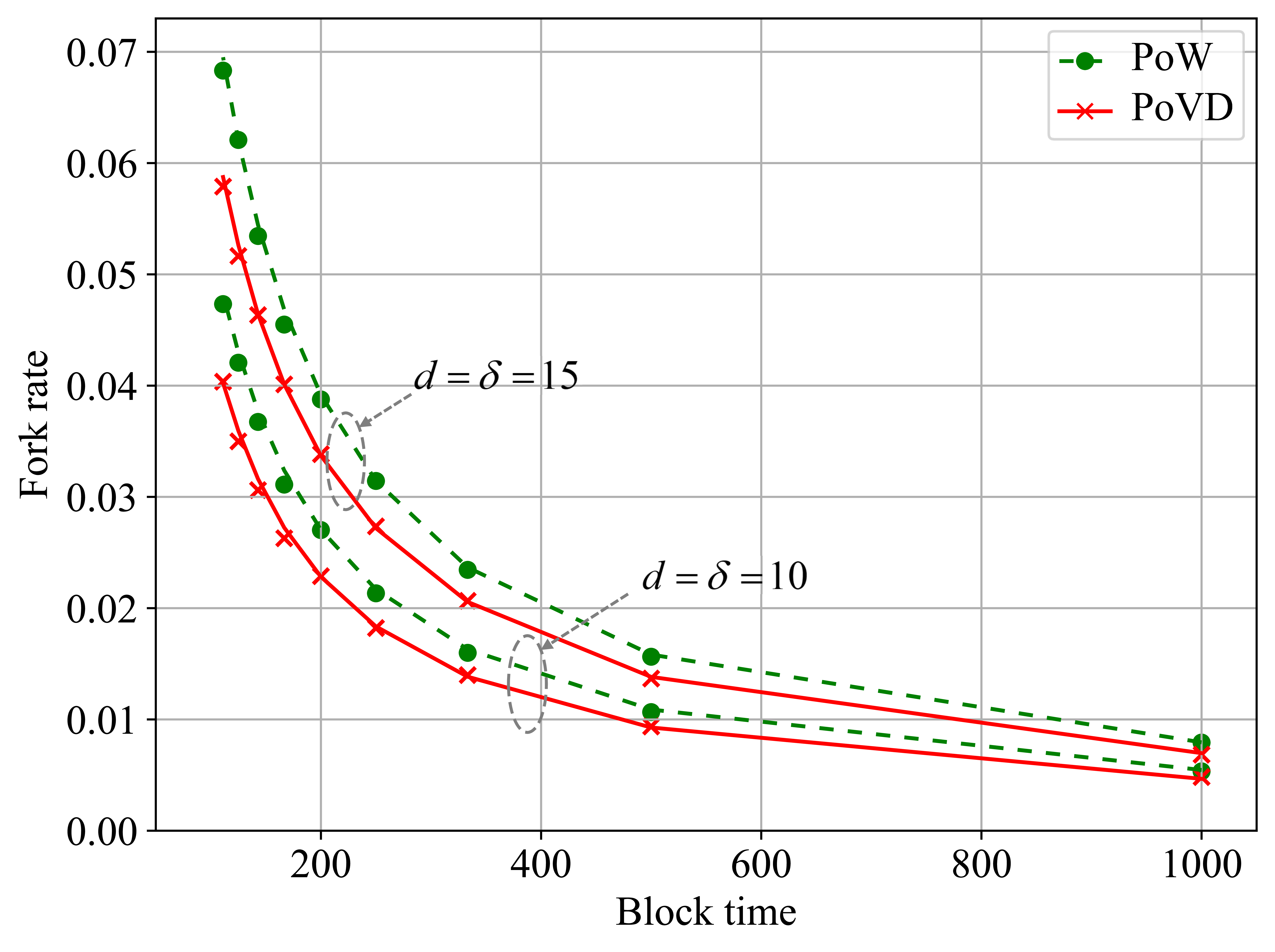}\protect
\par\end{centering}
\protect\caption{Fork rate of PoW and PoVD under different block time with $n=16$,
and $c=0.4$. (The curves are analytical results, and the markers
are experimental ones.) \label{fig:Fork rate simulation}}

\vspace{-0.4cm}
\end{figure}
}

In this section, we conduct experiments based on the ChainXim blockchain
simulator\footnote{ChainXim is a blockchain simulator developed by XinLab to simulate
and verify blockchain systems under different parameter settings,
and it is available on https://github.com/ChainXim-Team/ChainXim.} to demonstrate the communication- and computation-efficiency of PoVD
and validate PoVD's fork rate superiority. In the following figures,
PoVD, PoW, and PBFT refer to the corresponding protocols in this work,
\cite{Nakamoto2008} and \cite{Castro1999}. VDF baseline refers to
a class of VDF-based consensus protocols \cite{Long2019,Han2020,Mirkin2024}. 

\subsection{Communication and Computational Complexity}

First, we illustrate the communication and computational complexities
of different consensus protocols in Fig. \ref{fig:Comunication Complexity}
and Fig. \ref{fig:Computation Complexity}, respectively. In Fig.
\ref{fig:Comunication Complexity}, communication complexity is defined
as the average number of messages exchanged to generate a new block.
We can observe that, for all consensus protocols, the communication
complexity increases with the number of miners, but at different rates.
The communication complexity of PBFT increases significantly as the
number of miners grows. In contrast, PoVD, PoW, and the VDF baseline
are communication efficient. As the network scales, the advantage
of communication complexity becomes more pronounced, highlighting
PoVD's scalability for bandwidth-limited networks. 

Fig. \ref{fig:Computation Complexity} presents the computational
complexity of PoVD, PoW, and the VDF baseline, which is defined as
the average number of oracles performed by a miner to generate a block.
Specifically, in PoW, oracles refer to the computation of hashes,
while in VDF baseline and PoVD, they refer to the number of square
computations. In Fig. \ref{fig:Computation Complexity}, every miner
is equipped with a different number of processors, representing varying
power of parallelism. Fig. \ref{fig:Computation Complexity} demonstrates
that PoVD consistently maintains a low computational cost regardless
of the number of processors, whereas PoW and the VDF baseline exhibit
a significant increase. This difference in computational complexity
emphasizes PoVD's advantage in computational-efficiency compared to
PoW and other benchmarks, making it suitable for lightweight devices.

\subsection{Fork Rate Superiority}

Now, we present the fork rates under different conditions. VDF baselines
have similar performance to PoW and are thus omitted. In Fig. \ref{fig:n=00003D10, c=00003D0.4 and l=00003D10}
and Fig. \ref{fig:Fork rate simulation}, the network capability indicator
is set to satisfy $c\leq\frac{1}{2}-\frac{1}{2d}$. Both figures demonstrate
that the fork rate of PoVD is lower than that of PoW, which aligns
with Theorem \ref{thm:basic t1}. Moreover, Fig. \ref{fig:n=00003D10, c=00003D0.4 and l=00003D10}
shows that when the propagation delay $d$ is fixed, the advantage
of PoVD in fork rate becomes more significant as the propagation speed
decreases (i.e., as $c$ decreases). In contrast, Fig. \ref{fig:Fork rate simulation}
illustrates that when the network capability indicator $c$ remains
constant, increasing the propagation delay $d$ does not significantly
affect the consistency advantage but increases the fork rate.

We also present the fork rates across two special propagation networks
in Fig. \ref{fig:Fork rate simulation of two vector}, where the miners
are connected in a loop. In this case, we have the theoretical propagation
vector and thus derive the accurate fork rate and block time accordingly.
In Fig. \ref{fig:Fork rate simulation of two vector}, analytical
and simulation results are represented by lines and markers, respectively.
The resulting propagation vectors are $\mathbf{w}_{8}=\left(1/16,3/16,\ldots,15/16\right)$
for 16 miners, and $\mathbf{w}_{16}=\left(1/32,3/32,\ldots,31/32\right)$
for 32 miners. As expected, Fig. \ref{fig:Fork rate simulation of two vector}
shows that PoVD can achieve a lower fork rate than PoW, as it can
be easily verified that the network capability indicators corresponding
to $\mathbf{w}_{8}$ and $\mathbf{w}_{16}$ also satisfy the sufficient
condition $c\leq\frac{1}{2}-\frac{1}{2d}$.

Fig. \ref{fig:Fork rate simulation of l<d} shows the performance
of PoVD over a long-tail network with a large $d$. In the experiments,
we set $w_{18}\thickapprox0.97$ and $w_{19}\thickapprox0.99$ for
$d=20$ and $w_{46}\thickapprox0.98$ and $w_{48}\thickapprox0.99$
for $d=50$. Fig. \ref{fig:Fork rate simulation of l<d} implies that,
even when $d>\delta$, PoVD still exhibits a lower fork rate than
PoW. This supports the remarks given after Theorem \ref{thm:basic t1}.
Hence, the network condition for PoVD to achieve better performance
than PoW is quite weak.

To more intuitively demonstrate the fundamental differences in fork
rates between PoVD and PoW, we evaluate the concurrent fork rate and
the propagation fork rate under a 16-miner network in Fig. \ref{fig:Fork rate under different Target}.
In the legend, concurrent fork refers to the generation of multiple
blocks in the same round, and propagation fork refers to the block
conflicts during propagation. The fork rate of PoW consists of both
the concurrent and propagation fork rate, and the propagation fork
rate is the dominant component. In contrast, PoVD has no propagation
fork, which makes the total fork rate lower than PoW.

\cprotect\subsection{Experimental Energy Consumption
\begin{figure}
\protect\begin{centering}
\protect\includegraphics[width=0.45\textwidth]{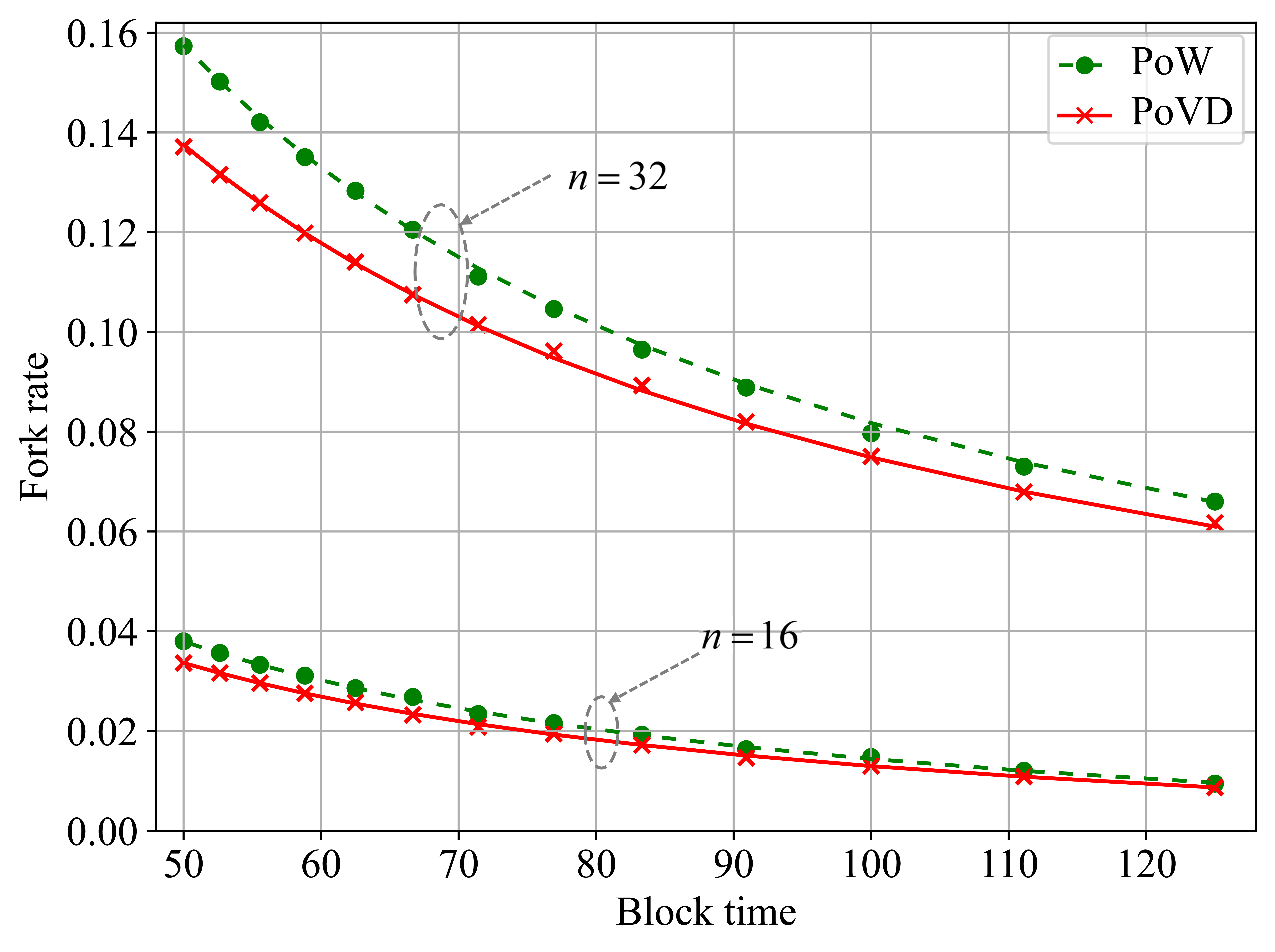}\protect
\par\end{centering}
\protect\caption{Fork rate of PoW and PoVD under circle topology with 16 miners, and
32 miners. (The curves are analytical results, and the markers are
experimental ones.) \label{fig:Fork rate simulation of two vector}}

\vspace{-0.4cm}
\end{figure}
\begin{figure}
\protect\begin{centering}
\protect\includegraphics[width=0.45\textwidth]{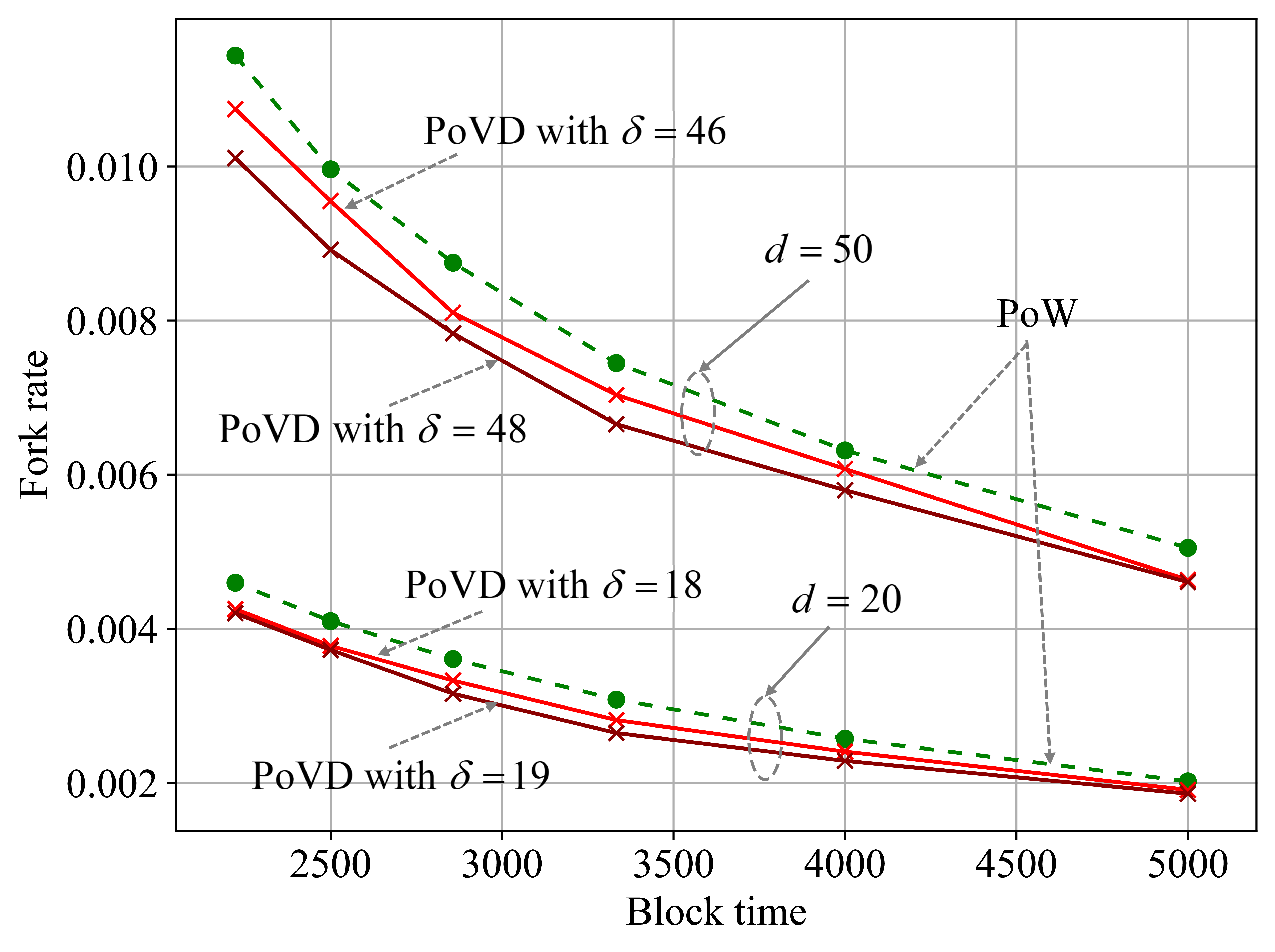}\protect
\par\end{centering}
\protect\caption{Fork rate of PoW and PoVD under different block time with 256 miners.
\label{fig:Fork rate simulation of l<d}}

\vspace{-0.4cm}
\end{figure}
\begin{figure}[t]
\protect\centering{}\protect\includegraphics[width=0.45\textwidth]{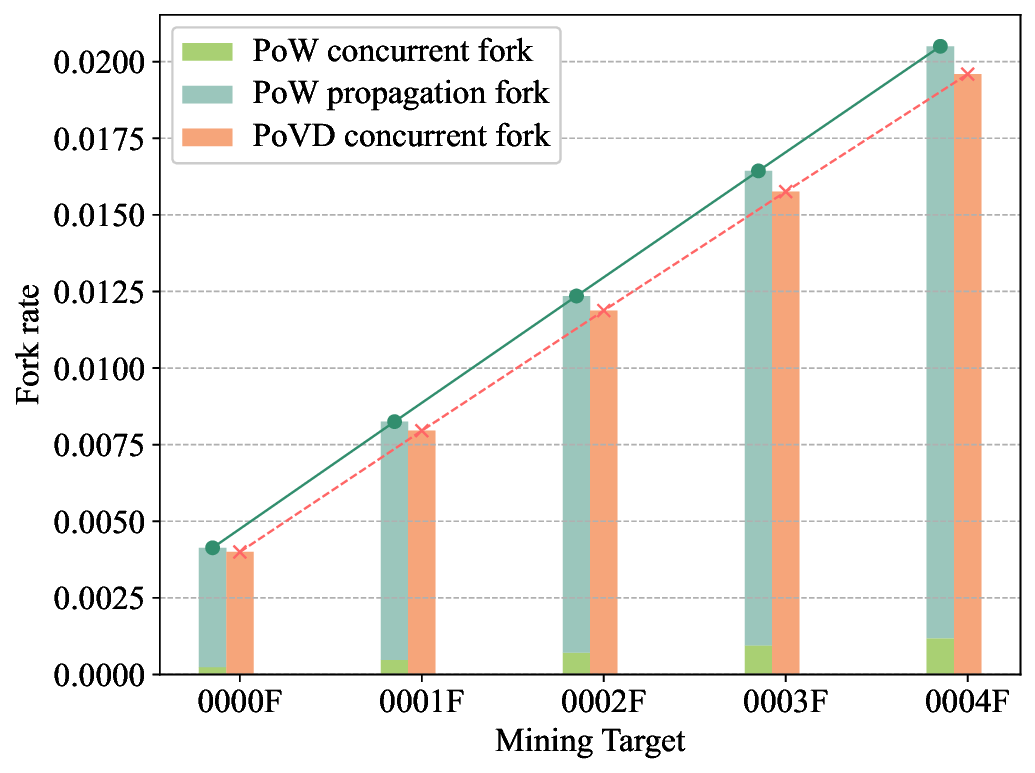}\protect\caption{Statistical results of different types of fork rate under varying
difficulty target with 16 miners. (PoVD has no propagation fork.)\label{fig:Fork rate under different Target}}
\vspace{-0.4cm}\protect
\end{figure}
\begin{figure}[t]
\protect\centering{}\protect\includegraphics[width=0.465\textwidth]{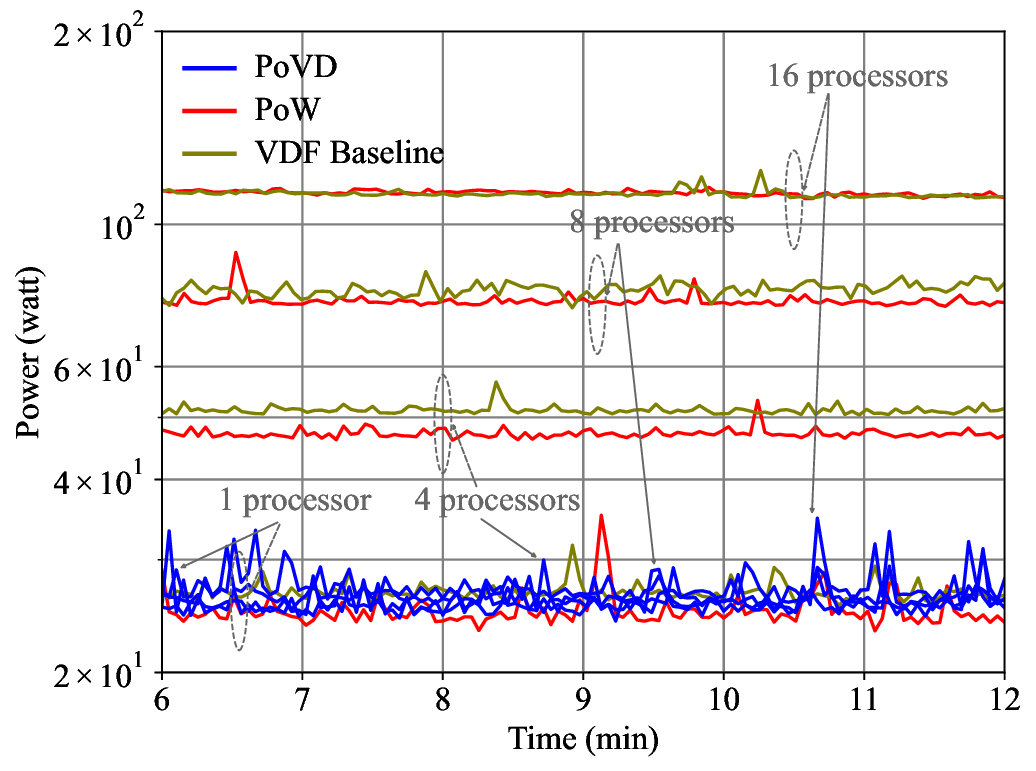}\protect\caption{Experimental energy consumption of different consensus protocols with
various numbers of processors.\label{fig:Energy consumption}}
\vspace{-0.4cm}\protect
\end{figure}
}

Finally, we measure the CPU power consumption of PoVD, PoW, and VDF
baseline through experiments. The experiments are based on a Windows
desktop computer with a 16-core 4.3 GHz CPU and 128 GB of RAM. Fig.
\ref{fig:Energy consumption} presents the real-time energy consumption
by using different numbers of processors. In the figure, the power
curves from light to dark correspond to 1, 4, 8, and 16 processors,
respectively. We can observe that the energy consumption of PoW and
VDF baseline with a single process is almost identical to that of
PoVD. However, as the number of processors increases, the energy consumption
of PoW and the VDF baseline grows significantly, while PoVD remains
at low power consumption. This comparison highlights that PoVD is
more energy-efficient than PoW and VDF baseline and is more suitable
for resource-constrained devices.

\section{\label{sec:Conclusions}Conclusions}

In this study, we proposed PoVD, an efficient consensus protocol
based on VDF. We designed the block generation and verification rules
of PoVD and discussed potential security implications with corresponding
countermeasures. Our design significantly reduces both communication
and computation overhead and makes PoVD suitable for establishing
consensus in resource-constrained and bandwidth-limited environments.
Furthermore, PoVD can control the block time distribution according
to, e.g., the network environment, and thus can more flexibly balance
the trade-off between blockchain consistency and throughput. Due to
the above reasons, we rethought the mining model of PoVD and derived
the fork rate and block time of PoVD under an adjustable block time
distribution. Through mathematical proof, we pointed out that, under
some relatively weak network conditions, PoVD can achieve a lower
fork rate than PoW for equivalent blockchain throughput. Finally,
experiments against PBFT, PoW, and other VDF-based benchmarks illustrated
the low computational and communication complexities of PoVD and highlighted
its ability to adjust block time distribution for improving blockchain
consistency. 

\bibliographystyle{IEEEtran}
\bibliography{IEEEabrv,refference}

\end{document}